\documentclass[12pt,draftclsnofoot,onecolumn]{IEEEtran}
\usepackage{winsnotation}
\usepackage{cite}
\usepackage{amsmath,amssymb,amsfonts}
\usepackage{algorithm}
\usepackage{algorithmic}
\usepackage{CJK}
\usepackage{longtable,booktabs}
\usepackage{booktabs}
\usepackage{graphicx}
\usepackage{subfigure}
\usepackage{textcomp}
\usepackage{flushend}
\usepackage{amsthm}
\usepackage{upgreek}
\usepackage{graphicx}
\usepackage{supertabular}
\usepackage[bottom]{footmisc}
\usepackage{epstopdf}

\newtheorem{Theorem}{\hskip\parindent\it{Theorem}}
\newtheorem{lemma}{\hskip\parindent\it{Lemma}}
\newtheorem{remark}{\hskip\parindent\it{Remark}}
\def\BibTeX{{\rm B\kern-.05em{\sc i\kern-.025em b}\kern-.08em
    T\kern-.1667em\lower.7ex\hbox{E}\kern-.125emX}}
\begin{document}
\title{Blind Channel Estimation and Data Detection with Unknown Modulation and Coding Scheme}
%\title{Blind Channel Estimation and Data Detection with Unknown Modulation and Coding Schemes in Multipath Channel}

\author{\IEEEauthorblockN{
Yu Liu
and Fanggang Wang
}
}

\maketitle

\begin{abstract}
This paper investigates a complete blind receiver approach in an unknown multipath fading channel, which has multiple tasks including blind channel estimation, noise power estimation, modulation classification, channel coding recognition, and data detection.
The side information required only is the candidates of the channel encoders and the modulation formats.
Each of these tasks has been sufficiently studied in the literature.
Few works studied the combination of two or three of them jointly.
However, to the best of our knowledge, this overall problem which involves the five aforementioned tasks has not been investigated previously.
Simply cascading the solution to each individual task naively is apparently far from the optimality.
This paper is the first attempt to address this overall problem jointly.
We propose a complete {\em blind} receiver approach that jointly {\em estimates} the unknown parameters (channel state information and noise power), {\em recognizes} the unknown patterns (modulation and coding scheme), {\em detects} the data of interest, and thus named {\em BERD} receiver.
In particular, the proposed BERD receiver exhibits an iterative manner, and the essential steps in the iteration are as follows:
1) multipath channel estimation based on the expectation-maximization algorithm;
2) noise power estimation;
3) received signal equalization using the Bayes equalizer;
4) soft-output demodulation and decoding;
5) re-encoding and re-modulation.
Another merit of the proposed BERD receiver is that it can be implemented for both cases of a single receiver and multiple receivers.
For multiple receivers, it supports both distributed and cooperative manners and allowing multiple receivers ensures successful estimation, recognition, and detection for such an extremely difficult problem.
Furthermore, the solution to the overall problem applies to any reduced one with parts of the five tasks.
The BERD receiver applies to the reduced problems as well and it still outperforms the exiting work on the individual or the joint tasks, which is validated by the simulation results.
In addition, numerical results show the performance of the complete blind BERD receiver within three folds:
a) Regarding estimation, the BERD receiver outperforms the linear minimum mean squared error (LMMSE) pilot-based channel estimator by over $3.5\,\mathrm{dB}$ at the mean square error of $10^{-2}$;
b) Regarding recognition, the correct modulation/coding recognition performance of the BERD receiver is within $0.3\,\mathrm{dB}$ as close to the recognition benchmark when the perfect channel state information (CSI) is available;
c) Regarding detection, the BERD receiver is within $0.5\,\mathrm{dB}$ at the bit error rate of $10^{-3}$ compared to the benchmark when the modulation, the channel coding, and the CSI are perfectly known.
Finally, the BERD receiver finds many applications in both civilian and military scenarios, such as the interference cancelation in spectrum sharing, real-time signal interception, and processing in electronic warfare operations, automatic recognition of a detect signal in software-defined radio, etc.
\end{abstract}

\section{Introduction}\label{S1}
With the rapid development of wireless communication, the increasing demands for high data rate, reliability, and quality of service (QoS) have attracted significant research attention.
The lack of spectrum resources due to the explosive data traffic becomes an urgent problem to be solved \cite{INTencoder4}.
Various standardization organizations have proposed flexible dynamic spectrum access and sharing technologies to improve the spectrum efficiency with {\em a priori} information of the spectrum occupation.
However, in a non-cooperative communication manner, a receiver is incapable of getting {\em a priori} information from the desired signals.
In addition, even in a cooperative manner, there are still co-channel interference from the adjacent cells, the transmission of other operators, and even some malicious emitters, which are cumbersome without any prior information of related parameters.
To address this issue, the techniques of the blind channel estimation, modulation classification, channel encoder identification, and blind data detection, etc., have emerged accordingly and played important roles in both the military and the civilian applications \cite{OA1}.
In this paper, we investigate a complete blind receiver approach, which is designed to estimate related parameters, recognize the unknown modulation and coding patterns, and detect the data of interest, with no {\em a priori}  information.

The overall blind receiver design is composed of the following five tasks, i.e., blind channel estimation, noise power estimation, modulation classification, channel encoder identification, and blind data detection.
Most of the individual tasks have been sufficiently studied in the literature.
Blind channel estimation has been studied in \cite{ref:CE-subspace-LB,ref:CE-ML-MMSE,ref:CE-ML-KL,
ref:CE-classes-sos,ref:CE-sos-timedomain,
ref:CE-multipath-subspace,ref:CE-multipath-modal,
ref:CE-MB-matching-correlation,ref:CE-MB-matching-cyc2,
ref:CE-MB-matching-cyc}, which can be classified into the maximum likelihood-based and moment-based methods \cite{ref:CE-subspace-LB}.
Modulation classification has been investigated in \cite{ref:modut1,ref:modut2,ref[TWC15],ref[TWC17],FC1,FC2,ref[TWC4],ref[TWC5],
MC-LB-ALRT-MILCOM,MC-LB-WEI,MC-LB-MILCOM,MC-LB-TSMC,
MC-LB-GLRT-BDULEK,MC-LB-GLRT-CCECE,MC-LB-HLRT-MIL2000,MC-LB-HLRT-TWC2015,
ref:MC-wcsp13,ref:MC-wcnc11,MC-FB-CL2011,
MC-FB-PDF,MC-FB-phase,MC-FB-high}, including both the likelihood-based (LB) methods and the feature-based (FB) methods.
The channel encoder identification has been studied in \cite{code1,CODEMS1,INTencoder4,TB3,ref:blind-turbo,ref:blind-cyclic,ref:EC-linear-block,ref:EC-cyclic-code,ref:EC-convo-code,ref:EC-turbo}, which can recognize different channel encoders, including both block codes and convolutional codes.
Few works studied the combination of two or three tasks of the five ones.
In \cite{ref:blind-DT1,ref:DT-CE-BW,ref:DT-MC,ref:ZJW,TWC[26],LDPC1,MULTI1,ref:LY}, two tasks were investigated jointly.
Blind data detection and channel estimation were simultaneously studied in \cite{ref:blind-DT1,ref:DT-CE-BW}.
Blind data detection and modulation classification were considered jointly to improve the data detection performance in \cite{ref:DT-MC}.
In \cite{ref:ZJW,TWC[26],LDPC1,MULTI1,ref:LY}, three tasks were considered at the same time.
In \cite{ref:ZJW,TWC[26]}, the blind channel estimation, modulation classification, and blind data detection were addressed jointly.
The joint approaches for blind channel estimation, noise power estimation, and encoder identification were also investigated in \cite{LDPC1,MULTI1}.
Recently, we proposed a joint scheme in \cite{ref:LY}, which simultaneously accomplished the channel estimation, encoder recognition, and data detection.
In the following, the literature of each individual task and the combination of partial tasks were reviewed respectively.

Regarding the two classes of the blind channel estimation in \cite{ref:CE-subspace-LB}, the maximum likelihood-based methods are usually optimal for big data records and they approach the minimum variance unbiased estimators, which has been investigated in \cite{ref:CE-ML-MMSE,ref:CE-ML-KL}.
Unfortunately, it is difficult to derive the closed-form solutions of the maximum likelihood-based methods since the existence of the local optimal solutions complicates the implementation of the methods.
In light of this, the moment-based methods are proposed, which can be further classified into the subspace approaches \cite{ref:CE-classes-sos,ref:CE-sos-timedomain,ref:CE-multipath-subspace,ref:CE-multipath-modal} and the moment matching approaches \cite{ref:CE-MB-matching-correlation,ref:CE-MB-matching-cyc2,ref:CE-MB-matching-cyc}.
In \cite{ref:CE-classes-sos}, the classes of the multipath channels were identified from the second-order statistics using multiple antennas.
A parametric subspace approach using the second-order moment was adopted to identify the specular multipath propagation channels in \cite{ref:CE-multipath-subspace}.
The proposed parametric method can estimate the channel parameters including attenuations, relative delays, and spatial signatures, which are robust to the channel order overestimation compared to the classical subspace method.
To achieve more robust performance against channel conditions and channel order selection, the moment matching methods were developed.
The cross-correlation matching approach based on the second-order statistics of the channel outputs was proposed in \cite{ref:CE-MB-matching-correlation}, which estimates the channel response without knowing the length of the finite impulse response channel.
In \cite{ref:CE-MB-matching-cyc2, ref:CE-MB-matching-cyc}, the cyclic correlation matching algorithms were investigated to estimate the channel impulse response and the variance of the additive noise.
Even when the channel is not uniquely estimated from the second-order statistics, the proposed approach still provides a useful estimate.
However, the moment matching methods are not easy to implement due to the multiple local optimal solutions and the cost of the complexity.

The modulation classification methods are categorized into two groups, i.e. the LB methods and the FB methods \cite{ref[TWC4],ref[TWC5]}.
The LB methods have been thoroughly investigated in the additive white Gaussian noise (AWGN) channel and the flat-fading channel, as the LB method is the optimal classifier in the Bayesian sense \cite{ref:modut1,ref:modut2}.
Regarding the model built for the unknown parameters, three prominent approaches have been proposed, i.e., average likelihood ratio test (ALRT)\cite{MC-LB-ALRT-MILCOM,MC-LB-WEI,MC-LB-MILCOM,MC-LB-TSMC}, generalized likelihood ratio test (GLRT) \cite{MC-LB-GLRT-BDULEK,MC-LB-GLRT-CCECE}, and hybrid likelihood ratio test (HLRT) \cite{MC-LB-HLRT-MIL2000,MC-LB-HLRT-TWC2015}.
However, the LB methods have high computational complexity and sensitivity to the unknown channel conditions.
In contrast, the FB methods have much lower complexity and could be robust to some particular conditions as per feature extraction \cite{FC1,FC2,MC-FB-CL2011,ref[TWC15],ref[TWC17],ref:MC-wcsp13,ref:MC-wcnc11}. A large amount of features has been proposed in the literature, such as the statistical moments and the probability density function (PDF) of the phase to classify the phase-shift keying (PSK) modulation \cite{MC-FB-PDF,MC-FB-phase}, and the cyclic cumulants for the high-order modulation classification \cite{MC-FB-high}, etc.
In \cite{FC1,FC2}, the higher-order statistics are applied to solve the classification task in the unknown multipath channel.
In \cite{FC1}, a blind channel estimator was proposed first, and then a fourth-order cumulant-based classifier is developed to extract essential features for classification.
However, the channel state information estimated from the fourth-order moments is inaccurate.
An enhanced approach using sixth-order cumulants is proposed in \cite{FC2} to improve the classification performance using this inaccurate channel information.
In our earlier work \cite{ref[TWC15],ref[TWC17],ref:MC-wcsp13,ref:MC-wcnc11}, a goodness of fit approach using Kolmogorov-Smirnov (KS) was proposed to solve the modulation classification in various channels, such as the AWGN channel, the flat-fading channel, and the channel with unknown phase and/or unknown frequency offsets.
The proposed algorithm achieves better classification performance and even lower complexity than the cumulant-based ones.
However, the feature extraction in the FB methods is difficult to be incorporated with the likelihood-based soft demodulation and decoding at a receiver. That is, the joint design is troublesome.

The channel encoder identification is to determine the unknown channel encoder from the output bits of demodulation.
The existing work of coding identification is mainly distinguished between two types of error-correcting codes, i.e., the block codes and the convolutional codes \cite{code1,CODEMS1,ref:blind-turbo,ref:blind-cyclic,INTencoder4,
TB3,ref:EC-linear-block,ref:EC-cyclic-code,ref:EC-convo-code,ref:EC-turbo}.
The linear block codes identification based on Euclidean distance distribution was studied in \cite{CODEMS1}, which determines both the code length and the code dimension from a soft output of demodulation.
In \cite{TB3}, a blind encoder identification for low-density parity-check (LDPC) codes as well as frame synchronization was investigated in the multipath fading channel.
A two-stage search method using the quasi-cyclic nature of the parity-check matrix was proposed.
In \cite{ref:blind-cyclic}, the blind reconstruction of the binary cyclic codes was discussed.
The proposed approach identifies the correct synchronization, the length, and the factors of the generator polynomial of the code.
An iterative method for the convolutional encoder identification at a specific coding rate was proposed in \cite{INTencoder4}.
The blind identification method based on the algebraic properties of the convolutional encoder was considered both in noiseless and noisy cases.
A turbo encoder identification scheme based on the expectation-maximization (EM) algorithm was studied in \cite{ref:blind-turbo}, which can determine the optimal connections of the shift-registers.
Moreover, a joint identification scheme for the type of error-correcting codes and the interleaver parameters was studied in \cite{code1}.
The proposed scheme classifies the incoming data among block codes, convolutional codes, and uncoded data based on the analytical and histogram approaches.

In addition to the previous work mostly focused on an individual task, the joint problems by combing several of these tasks have been investigated recently.
In \cite{ref:blind-DT1,ref:DT-CE-BW}, the channel estimation and data detection were simultaneously studied.
A Bayes equalizer was designed for the restoration of finite-alphabet symbols and the Gibbs sampler was adopted to estimate the complex coefficients of both the Gaussian intersymbol interference (ISI) channel and the non-Gaussian ISI channel in \cite{ref:blind-DT1}.
To improve data detection performance, the modulation classification and data detection were jointly investigated in \cite{ref:DT-MC}.
The proposed method improved the symbol detection performance via relaxing the constraints on the modulation classification performance in the AWGN channel.
Blind channel estimation, modulation classification, and data detection were jointly considered \cite{ref:ZJW,TWC[26]}.
An LB scheme was proposed in \cite{ref:ZJW}, which jointly estimates the multipath channel and classifies the unknown modulation formats.
In \cite{TWC[26]}, a hybrid maximum likelihood modulation classification scheme using the EM algorithm was proposed.
The method blindly estimates unknown time offset, channel amplitude, and channel phase in a flat-fading channel.
In \cite{LDPC1,MULTI1}, the tasks of blind channel estimation, noise power estimation, and channel encoder identification are investigated jointly.
In \cite{LDPC1}, the unknown LDPC encoder is identified by using the average log-likelihood ratio (LLR) of the {\em a posteriori} probability (APP) of the syndrome, where the unknown channel gain and the noise power are estimated by the EM algorithm.
In \cite{MULTI1}, a blind LDPC encoder identification scheme was firstly proposed for quadrature amplitude modulation (QAM) signals in a flat-fading channel.
The EM algorithm was adopted as well to estimate channel amplitude, channel phase, and noise power.
Recently, we proposed a joint channel estimation, encoder identification, and data detection scheme in \cite{ref:LY}.
The proposed approach iterates between an EM-based channel estimator and a Bayes detector, which simultaneously estimates the channel gain, channel phase, and recognize the channel coding.
In summary, the aforementioned literature focuses on either the individual task or the combination of two or three tasks of the overall problem in this paper.
To the best of our knowledge, the overall problem consisting of all the five tasks has not been addressed previously in the literature.
A straightforward recipe is to simply cascade the solutions to each individual task, which is apparently far from the optimal solution.

In this paper, we make a first attempt to consider the overall problem and propose a complete blind receiver approach, which jointly estimates the channel state information and the noise power, recognizes the unknown modulation and coding scheme (MCS), detects the data of interest, and thus is called BERD receiver.
Regarding the difficulty of inter-symbol interference induced by the multipath channel, the BERD receiver is well designed and exhibits an iterative manner among different modules with each hypothesis candidate of MCS, i.e., the blind channel and noise power estimator, the Bayes equalizer module, the soft demodulator and decoder module, the re-encoder and re-modulator module, the stop criteria module, and the multistage likelihood decision module.
The essential steps in the iteration for each hypothetical candidate modulation and channel coding scheme are summarized as follows:
1) the EM algorithm is applied to estimate the unknown multipath channel including both the amplitude and the phases of each path;
2) the noise power is determined simply by subtracting the noise-free signal reconstructed by the estimated channel states and the information symbols predetermined in the previous iteration from the received signal;
3) given the estimated channel state information and the noise power, the received signal is equalized by using the Bayes equalizer to obtain the {\em a posteriori} probability of each modulated symbol;
4) the soft-output symbols from the equalizer is demodulated and decoded for each candidate MCS (as the final decision is made out of the iteration);
5) the output bits of the decoder is re-encoded and re-modulated with the corresponding MCS, which is required in the step of channel estimation in the next iteration.
The main contributions of this paper are summarized as follows:
\begin{itemize}
\item 	Proposed a complete blind receiver approach in a multipath fading channel, i.e., the BERD receiver, which solves the five tasks jointly, including blind channel estimation, noise power estimation, modulation classification, channel coding identification, and data detection. To the best of our knowledge, the BERD receiver is the first attempt to investigate the overall problem, which iteratively proceeds each of the five tasks.

\item Design a soft-information detector to iteratively enhance the accuracy of channel estimation and the correctness of data detection when the MCS is unknown. The detector contains a Bayes equalizer, a soft demodulator, and a soft decoder. The main advantage is that errors are corrected and then more reliable modulated symbols are regenerated for future channel estimation.
    The accuracy of channel estimation is improved accordingly which further helps the following detection. The iterative approach provides an efficient solution to the joint problem in a multipath fading channel.

\item 	The proposed BERD receiver is applicable to both single and multiple receivers. For a single receiver, the classification performance and the BER can be improved by allowing more iterations, while using multiple receivers cooperatively facilitate a shorter delay since fewer iterations are required to achieve an identical performance. Furthermore, the BERD receiver also supports a distributed manner that the decision of each receiver is fused at the end instead of soft likelihood information fusion during iteration.

\item 	The proposed BERD receiver is dedicated to the case of linear block codes. However, it can be easily extended to the other channel codes having distinguishable features that can be characterized by the likelihood representation. Last but not least, the solution to the overall problem can be applied to any reduced version of the original problem, such as the individual task or any partial combination of the tasks.
\end{itemize}

The remainder of this paper is organized as follows.
Section \ref{sec:system-model} introduces the system model.
The proposed BERD receiver is presented in Section \ref{sec:joint-scheme} and the solution to each individual task is addressed in the following sections.
In Section \ref{sec:channel_and_noise}, the blind channel estimation and the noise power estimation are proposed.
The soft-information detector is studied in Section \ref{sec:DEC}.
A multistage likelihood decision procedure is illustrated in Section \ref{sec:multistage-decision}.
The BERD approach for the system with multiple receivers is investigated in section \ref{sec:BERD-K-receive}.
Numerical results are shown in Section \ref{sec:simulation-results}.
At last, Section \ref{sec:conclusion} concludes this paper.

\textit{Notation:}
Throughout this paper, variables, vectors, and matrices are written as italic letters $x$, bold italic letters $\boldsymbol x$, and bold capital italic letters $\boldsymbol{X}$, respectively;
A random variable and its realization are respectively denoted by $\sf x$ and $x$;
a random matrix and its realization are respectively
denoted by $\boldsymbol{\sf X}$ and $\boldsymbol X$;
$|\mathcal{X}|$ is the cardinality of set $\mathcal{X}$;
$p(x)$ denotes the PDF $p_{\sf x}(x)$ of the random variable $\sf x$, and $p(x|y)$ denotes the conditional PDF $p_{\sf{x|y}}(x|y)$ of the random variable $\sf x$ conditioned on random variable $\sf y$;
$\mathbb{E}\{\cdot\}$ denotes the expectation with respect to (w.r.t.) all the randomness in the argument;
$\Re\{c\}$ and $\Im\{c\}$ represent the real and imaginary part of the complex number $c$, respectively;
$\mathcal{CN}(\mu,\sigma^2)$ denotes the PDF of a random variable following the complex Gaussian distribution with mean vector $\mu$ and the variance $\sigma^2$;
$\mathrm{GF}(q)$ denotes the Galois field of the integer $q$.
The operators $[\cdot]^{\mathrm{T}}$, $[\cdot]^{*}$, $(\cdot)^{\mathrm{H}}$ denote the transpose, the conjugate, and the Hermitian of their arguments, respectively;
the operator $\|\cdot\|$ denotes the $\ell_2$ norm of the argument;
the operator $\lfloor\cdot\rfloor$ represents the floor of the argument;
the operator $\oplus$ represents the addition in $\mathrm{GF}(2)$ of the argument;
the operator $\otimes$ represents the Kronecker product of the argument;
the $L$-by-$L$ identity matrix and $L$-by-$1$ identity vector are denoted by $\boldsymbol{I}_L$ and $\boldsymbol{1}_L$, respectively;
$\log c$ and $\ln c$ denote the logarithm of a real number $c$ to the base $2$ and $e$, respectively;
the imaginary unit is denoted by $\imath=\sqrt{-1}$;
$\mathbb{Z}$, $\mathbb{Z}_2$, and $\mathbb{Z}_2^n$ represent all the integers, the set with $\{0,1\}$, and the set with $n$ elements which take the value from $\mathbb{Z}_2$, respectively.
Define $\mathcal{I}_K=\{1,2,\ldots,K\}$ as shorthand as the index set.
The definition of the notations is summarized in Table \ref{table:symbol-list} in Appendix \ref{app:table} for the convenience of the readers.

\section{System Model}
\label{sec:system-model}
Consider a non-cooperative wireless transmission in which the receiver has no prior knowledge of the multipath channel state information, the noise power, or the MCS scheme. The ultimate goal is to correctly decode the message of interest from the unknown signal.
To accomplish this task, it is required to estimate the multipath channel and the noise power without any aid of pilots, classify the unknown modulation $\eta\in\mathcal{M}$, recognize the unknown channel coding $\zeta\in\mathcal{C}$, and detect the data of interest.
Denote the MCS by $\theta=\{\eta, \zeta\}\in \mathcal{M}\times \mathcal{C}$.
Then, the received signal can be expressed as
\begin{equation}
\label{LY1}
{\mathsf r_{j}} = \sum\limits_{\ell\in\mathcal{I}_{L}-1} {{a_{\ell}}{e^{\imath{\varphi}_{\ell}}}{\sf s}_{j-\ell} + {{\sf v}_{j}}},
\quad j\in\mathcal{I}_N
 \end{equation}
where $L$ is the number of paths of
the wireless channel;\footnote{To simplify the notation, we start from the case of a single receiver for brief illustration, the notation for multiple receivers in both cooperative and distributed manners is defined later in Section \ref{sec:BERD-K-receive}. In addition, the channel fading $a_\ell$ at some spots of the delay profile could be zero since the number of channel paths $L$ is unknown in the blind communication system.} ${a}_{\ell}\geqorig0$ and ${\varphi}_{\ell}\in[0,2\pi)$ are the unknown channel gain and the phase of the $\ell$th path;
${\sf s}_{j}$ is the modulated symbol from the unknown constellation $\mathcal{S}^\eta$, which is the set of all constellation points in the modulation format $\eta$, and ${\sf s}_j$ maps to $\log |\mathcal{S}^\eta| $ coded bits in a codeword $\boldsymbol{\tilde{\sf c}}\in\mathbb{Z}_{2}^{n}$.
We first define the uncoded information bit sequence with the length of $q$ as $\boldsymbol{\sf b}\in\mathbb{Z}_{2}^q$.
Assume that a $(n,q)$ linear block code with code rate $R=\frac{q}{n}$ is adopted in the transmission.
The codeword $\boldsymbol{\tilde{\sf c}}$ is obtained by encoding $\boldsymbol{\sf b}$ using the generator matrix $\boldsymbol{G}\in\mathbb{Z}_2^{q\times n}$, which can be expressed as
\begin{align}
\boldsymbol{\tilde{\sf c}}=\boldsymbol{G}^{\mathrm T}\boldsymbol{\sf b}.
\end{align}
This generator matrix $\boldsymbol G$ corresponds to a unique parity-check matrix $\boldsymbol H\in\mathbb{Z}_{2}^{(n-q)\times n}$.
The relationship between them can be written as
\begin{align}
\boldsymbol H {\boldsymbol G}^{\mathrm T}=\boldsymbol 0.
\end{align}
The noise ${\sf v}_{j}$, $j\in\mathcal{I}_N$, follows independent identically distributed (i.i.d.) zero-mean circularly symmetric complex Gaussian (CSCG) distribution i.e., ${\sf v}_{j}\sim \mathcal{CN}(0,\sigma^{2})$, $j\in\mathcal{I}_N$.

The tasks of the proposed BERD receiver are to jointly estimate the multipath channel states, including channel gain ${a}_{\ell}$ and channel phase ${{\varphi}}_{\ell}$, $\ell\in\mathcal{I}_{L}-1$, estimate the noise power ${\sigma}^2$, determine the unknown modulation $\eta$ and the unknown  channel coding $\zeta$ from a candidate set $\mathcal{M}\times\mathcal{C}$, and the last but the most important, detect the transmitted information bits $\boldsymbol{\sf b}$.
In the following, we present the function of each module involved in the BERD receiver.

\section{The Proposed BERD Receiver}
\label{sec:joint-scheme}
In this section, the process of the proposed receiver is briefly exhibited by introducing each functional module, which is followed by the pseudo-code of the overall receiver algorithm.
The algorithm and the information flow of the proposed receiver are shown in Figure \ref{fig:system_model}.
The BERD receiver is composed of six modules,
i.e., the blind channel and noise power estimator, the Bayes equalizer module, the soft demodulator and decoder module, the re-encoder and re-modulator module, the stop criteria module, and the multistage likelihood decision module.
The function of the six modules can be summarized as follows.
The overall receiver algorithm is summarized in Algorithm \ref{algo:total}.
\begin{figure}
\centering
\subfigure[ ]
{\begin{minipage}[b]{0.75\textwidth}
\includegraphics[width=1\textwidth]{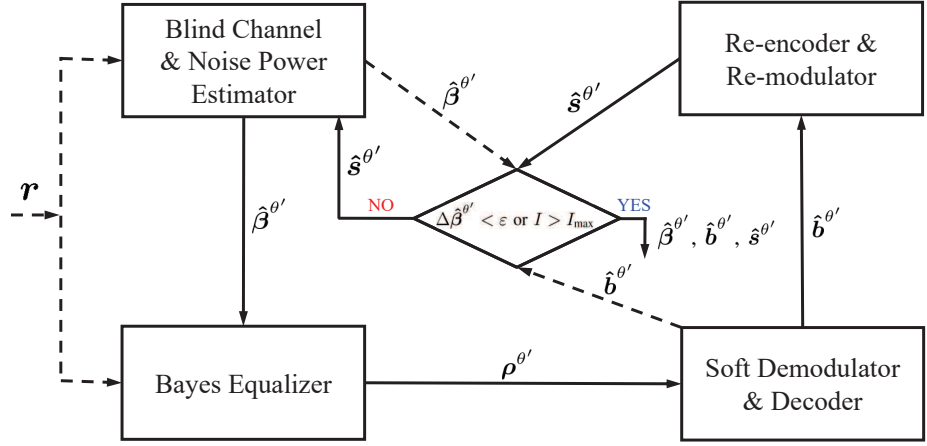}
\end{minipage}}
\subfigure[ ]
{\begin{minipage}[b]{0.5\textwidth}
\includegraphics[width=1\textwidth]{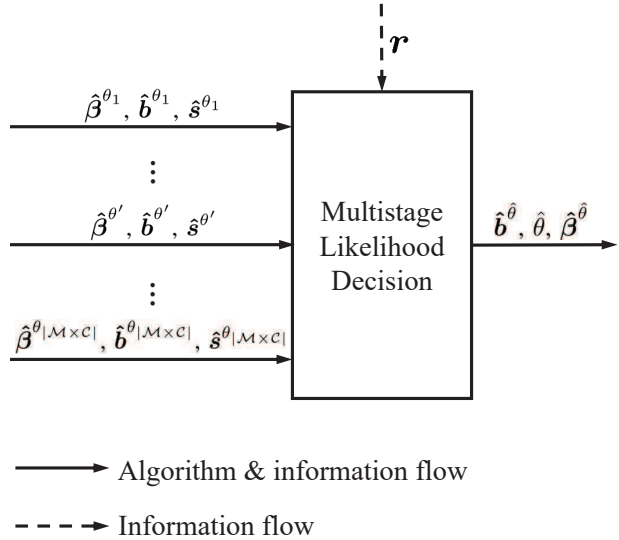}
\end{minipage}}
\caption{The block diagram of the BERD receiver.
(a) The algorithm and information flow between the different  modules in the hypothesis MCS candidate $\theta^{\prime}$.
(b) The algorithm and information flow of the multistage likelihood decision module.}
\label{fig:system_model}
\end{figure}

\noindent\emph{A. Blind Channel and Noise Power Estimator:}
{The proposed estimator is deployed to estimate the channel gain ${\boldsymbol {a}}$, the channel phase ${\boldsymbol {\varphi}}$, and the noise power ${{\sigma}^2}$.
In the hypothesis MCS scheme $\theta^{\prime}\in\mathcal {M}\times\mathcal C$, the inputs of the estimator is the received signal $\boldsymbol{\sf r}$ and the re-modulated symbols $\boldsymbol{\hat {s}}^{\theta^{\prime}}$, which are regenerated by the following re-encoder and re-modulator module.
The outputs of the proposed estimator are the channel state information ${\hat a_\ell}$ and $\hat \varphi_\ell$, $\ell\in\mathcal I_L-1$, and the noise power $\hat \sigma^2$, which are collectively denoted by
$\boldsymbol{\hat{\beta}}^{\theta^{\prime}}=[{\hat a_0, \hat a_1,\ldots,\hat a_{L-1}},{\hat \varphi_0,\hat \varphi_1,\ldots,\hat \varphi_{L-1}},{{\hat \sigma}^2}]^{\mathrm T}$.
The details of the blind channel and noise estimator will be further illustrated in Section \ref{sec:EM-based}.}

\noindent\emph{B. Bayes Equalizer:}
{The Bayes equalizer is adopted to equalize the multipath effect of the wireless channel.
The inputs of the equalizer are the received signal $\boldsymbol{\sf r}$ and the estimated channel information $\boldsymbol{\hat{\beta}}^{\theta^{\prime}}$.
The output of the Bayes equalizer is the posterior probability of the modulated symbols $\boldsymbol\rho^{\theta^{\prime}}$, which serves as the input of the soft demodulator and decoder module.
The details of the Bayes equalizer are deferred to Section \ref{sec:DEC}.}

\noindent\emph{C. Soft Demodulator and Decoder Module:}
{This module is applied to demodulate and decode the output  signal from the Bayes equalizer.
It suppresses the noise and the inter-symbol interference induced by the multipath channel.
The input of this module is the output of the Bayes equalizer,
and the output of it is the decoded bits $\boldsymbol{\hat b}^{\theta^{\prime}}$.
The details of the soft demodulator and decoder module are provided in Section \ref{sec:DEC}.}

\noindent\emph{D. Re-encoder and Re-modulator Module:}
{This module is deployed to re-encode and re-modulate the information bits from the previous soft demodulator and decoder module.
The input of this module is the decoded bits $\boldsymbol{\hat b}^{\theta^{\prime}}$.
The output is the regenerative modulated symbols $\boldsymbol{\hat s}^{\theta^{\prime}}$.
The details of the re-encoder and re-modulator module are presented in Section \ref{sec:DEC}.
In addition, the Bayes equalizer module, the soft demodulator and decoder module, the re-encoder and re-modulator module are cascaded to detect and regenerate the received signal in each iteration of the proposed BERD receiver.}

\begin{algorithm}[!t]
\caption{The Proposed BERD Receiver}
\begin{algorithmic}[1]
\label{algo:total}
\STATE{\textbf{Init:} $\boldsymbol{\hat\beta}^{\theta^{\prime}}$;}
\WHILE{the variation of the estimated channel is more than $\varepsilon$ or the number of iterations does not exceed the maximum threshold}
\STATE{Compute $\boldsymbol\rho^{\theta^{\prime}}$ in the Bayes equalizer module according to Section \ref{sec:DEC};}
\STATE{Detect $\boldsymbol{\hat b}^{\theta^{\prime}}$ in the soft demodulator and decoder module according to Section \ref{sec:DEC};}
\STATE{Regenerate $\boldsymbol{\hat s}^{\theta^{\prime}}$ in the re-encoder and re-modulator module according to Section \ref{sec:DEC};}
\STATE{Update $\boldsymbol{\hat\beta}^{\theta^{\prime}}$ in the blind channel and noise power estimator according to Section \ref{sec:EM-based};}
\ENDWHILE
\STATE{The outputs including the detected bits $\boldsymbol{\hat b}^{\hat \theta}$, the MCS $\hat\theta$, and the estimated channel information $\boldsymbol{\hat \beta}^{\hat \theta}$ are determined in the multistage likelihood decision module according to Section \ref{sec:multistage-decision};}
\end{algorithmic}
\end{algorithm}

\noindent\emph{F. Stop Criteria:}
{The stop criteria module decides whether the iteration stops.
The inputs of the stop criteria module are the estimated channel information $\boldsymbol{\hat{\beta}}^{\theta^{\prime}}$, the decoded bits $\boldsymbol{\hat b}^{\theta^{\prime}}$, and the regenerated symbols $\boldsymbol{\hat s}^{\theta^{\prime}}$ in the current iteration.
The stop criteria are that the mean square error (MSE) of the estimated channel information in the current iteration and the previous one is less than the stopping threshold $\varepsilon$, i.e.,
$\Delta \boldsymbol{\hat{\beta}}^{\theta^{\prime}}={\Big\|\boldsymbol{\hat{\beta}}^{\theta^{\prime}}[I+1]-\boldsymbol{\hat{\beta}}^{\theta^{\prime}}[I]\Big\|^2}<\varepsilon$
or the iterations exceed the maximum iterations, i.e., $I>I_{\text{max}}$.
If the stop criteria are not satisfied, the output $\boldsymbol{\hat s}^{\theta^{\prime}}$ is adopted by the blind channel and noise power estimator in the next iteration.
Otherwise, the iteration stops and this module outputs $\boldsymbol{\hat{\beta}}^{\theta^{\prime}}$, $\boldsymbol{\hat{b}}^{\theta^{\prime}}$, and $\boldsymbol{\hat s}^{\theta^{\prime}}$, which serve as the inputs of the following multistage likelihood decision module.}

\noindent\emph{G. Multistage Likelihood Decision Module:}
{This module makes the final decision of the information bit  $\boldsymbol{\hat b}^{\hat\theta}$, the adopted MCS $\hat\theta$, and the estimated channel information $\boldsymbol{\hat\beta}^{\hat\theta}$.
The inputs are the outputs of the stop criteria module in each hypothesis MCS candidate $\theta^{\prime}\in\mathcal{M}\times\mathcal{C}$.
The details of the multistage likelihood decision module are illustrated in Section \ref{sec:multistage-decision}.}

\section{Blind Channel and Noise Power Estimator}
\label{sec:channel_and_noise}
In this section, we propose an algorithm to estimate the unknown multipath channel information and the noise power, including the channel gain ${a}_\ell$, the channel phase ${{\varphi}}_\ell$, $\ell\in\mathcal{I}_{L}-1$, and the noise power ${\sigma}^2$.
To solve this problem, the maximum likelihood (ML) estimator is adopted, which aims to estimate the unknown parameter $\boldsymbol \beta=[{a_0,a_2,\ldots,a_{L-1}}, {\varphi_0,\varphi_1,\ldots,\varphi_{L-1}},{{\sigma}^2}]^{\mathrm T}$ of the likelihood function ${p}( {{\boldsymbol{r}}|\boldsymbol{\hat{s}};\boldsymbol{{\beta}}})$.\footnote{The modulated symbols $\boldsymbol{\hat s}$ can be obtained from the re-modulator and re-encoder module, which is in Section \ref{sec:DEC}.}
Then, for each hypothesis MCS candidate {$\theta^{\prime} \in \mathcal{M}\times\mathcal{C}$}, the explicit expression of ${p}( {{\boldsymbol{r}}|\boldsymbol{\hat{s}};\boldsymbol{{\beta}}})$ is given by
\begin{align}
{p}( {{\boldsymbol{r}}|\boldsymbol{\hat{s}};\boldsymbol{{\beta}}})
&=\prod_{j\in\mathcal I_N}{p}( {{r_j}|\boldsymbol{\hat{s}}_{j-L+1}^{j};\boldsymbol{{\beta}}}) \\
& \propto \frac{1}{{\sigma}^N}\exp{\Bigg(\sum_{j\in\mathcal{I}_N}-\frac{1}{{{\sigma}^2}}
{{\bigg|{r}_{j}-\sum_{\ell\in\mathcal{I}_{L}-1} {{a}_{\ell}{e^{\imath{\varphi}_{\ell}}}} {{\hat s}_{j-\ell}}\bigg|^{2}}}\Bigg)}
\label{eq:likelihood_probability}
\end{align}
where $\boldsymbol{\hat{s}}=[{\hat{s}}_1,{\hat{s}}_2,\ldots,{\hat{s}}_N]^{\mathrm T}\in\mathcal{S}^N$, and $\boldsymbol{\hat{s}}_{j-L+1}^j=[{\hat s}_{j-L+1},{\hat s}_{j-L+2},\ldots,{\hat s}_{j}]^{\mathrm T}\in\mathcal{S}^L$.\footnote{Note that, if $j\leqorig \ell$, ${\hat s}_{j-\ell}=0$. Considering the memory characteristics of the multipath channel, the received symbol $r_j$ is conditional independent to the other received symbols given $\boldsymbol{\hat {s}}_{j-L+1}^j$ and $\boldsymbol \beta$.}
Consequently, the log-likelihood function $\mathcal F(\boldsymbol {\beta})$ can be expressed as\footnote{The likelihood function $\mathcal F(\boldsymbol\beta)$ in \eqref{Lbeta} is evaluated in the multistage likelihood decision module, which is introduced in Section \ref{sec:multistage-decision}.}
\begin{align}
\mathcal F(\boldsymbol {\beta}) &= \ln{{p}( {{\boldsymbol{r}}| \boldsymbol{\hat{s}};\boldsymbol{{\beta}}})} \label{eq:likelihood_prob_origin} \\
& =\sum_{j\in\mathcal{I}_N}{-\frac{1}{{{\sigma}^{2}}}\bigg|{r}_{j}-\sum_{\ell\in\mathcal{I}_{L}-1} {{a}_{\ell}{e^{\imath{\varphi}_{\ell}}}} {\hat{s}_{j-\ell}}\bigg|^{2}}-N\ln{\sigma}.
\label{Lbeta}
\end{align}
Then, the maximum likelihood estimator (MLE) of $\boldsymbol {\beta}$ is given by
\begin{equation}\label{LYbeta}
\boldsymbol {\hat{\beta}}
= \mathop {\arg \max \limits_{\boldsymbol {\beta}}}  {\mathcal F}(\boldsymbol {\beta} ).
\end{equation}
However, the problem in \eqref{LYbeta} is a non-convex problem that is intractable.
In addition, due to the multipath scenario, the received signal is a superposition of the signals from all paths which are difficult to be decoupled.
In the following, to deal with this problem, we design a blind channel and noise power estimator to estimate the unknown parameter $\boldsymbol {\beta}$.

We proposed an EM-based estimation algorithm to provide a local optimal solution to \eqref{LYbeta}.
Assume $\hat s_{j}$ is the $j$th detected modulated symbol, and the total power of the transmitted symbols is $P=\sum_{j\in\mathcal{I}_{N}} |\hat s_{j}|^{2}$.
Additionally, define $\boldsymbol{\hat Z}[t]$ as the complete data and ${\hat z}_{\ell,j}[t]$ is the $j$th complete data of the $\ell$th path in iteration $t$.\footnote{The choice of the complete data ${\hat z}_{\ell,j}[t]$ has a significant impact on the convergence result of the EM-based algorithm, which can be derived according to \eqref{LYcomplete}.
The determination details of ${\hat z}_{\ell,j}[t]$ are discussed in Appendix \ref{app:EM}. In addition, it is noteworthy that $t$ is the iteration index of the EM-based algorithm.}
Then, the closed-form expressions of the estimated channel information are stated in the following Lemmas \ref{le:a_phi} and \ref{le:sigma}.
\label{sec:EM-based}
\begin{figure}[!t]
\centering
\includegraphics[width=5.5in]{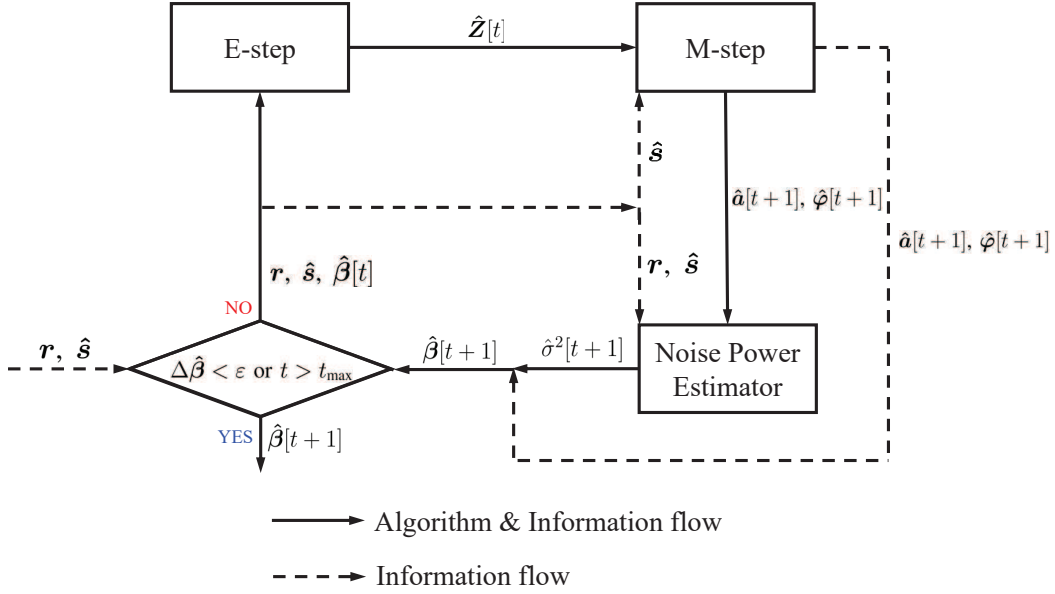}
\caption{The structure of the blind channel and noise power estimator.}
\label{fig:EM_module}
\end{figure}
\begin{lemma}
\label{le:a_phi}
\em{Given the modulated symbols ${\hat s}_j$, $j\in\mathcal I_N$, in the EM algorithm, the estimated channel gain ${\hat a}_{\ell}[t+1]$ and the estimated channel phase ${\hat\varphi}_{\ell}[t+1]$ in iteration $t+1$ are updated by}
\begin{equation}
\label{LYat}
{\hat a}_{\ell}[t+1] = \frac{1}{{{P}}}\sum\limits_{j\in\mathcal{I}_{N}} {\Re \big(
{{\hat s_{j-\ell}^{\ast}}{\hat{z}_{\ell,j}[t]}{e^{ - \imath{\hat\varphi}_{\ell}[t+1]}}} \big)}, \quad \ell\in\mathcal{I}_{L}-1
\end{equation}
\begin{equation}
\label{LYphit1}
{\hat\varphi}_{\ell}[t+1] = {\tan ^{ - 1}}\frac{\Im
 \Big( \sum_{j\in\mathcal{I}_N}{\hat s_{j-\ell}^{\ast}}{\hat{z}_{\ell,j}[t]} \Big)}
 {{\Re \Big( \sum_{j\in\mathcal{I}_N}{\hat s_{j-\ell}^{\ast}}{\hat{z}_{\ell,j}[t]} \Big)}},
 \quad \ell\in\mathcal{I}_{L}-1.
\end{equation}
\end{lemma}

\makeatletter
\renewenvironment{proof}[1][\proofname]{\par%
\pushQED{\qed}%
\normalfont \topsep6\p@\@plus6\p@\relax%
\trivlist%
\item[\hskip\labelsep%
#1]\ignorespaces%
}{%
\popQED\endtrivlist\@endpefalse%
}
\makeatother
\begin{proof}[\quad\it{Proof:}]
See Appendix \ref{app:EM}.
\end{proof}

\begin{lemma}
\label{le:sigma}
\em{Given the estimated channel gain ${\hat a}_\ell[t+1]$ in \eqref{LYat} and the channel phase ${\hat \varphi}_\ell[t+1]$ in \eqref{LYphit1}, $\ell\in\mathcal I_{L}-1$, the updated noise power ${\hat\sigma}^2[t+1]$ in iteration $t+1$ is given by}
\begin{align}
\label{sigma}
{\hat\sigma}^2[t+1] = \frac{1}{N}\sum_{j\in\mathcal{I}_{N}}\bigg|r_{j} - \sum_{\ell\in\mathcal{I}_{L}-1}{\hat a}_{\ell}[t+1]e^{\imath{\hat\varphi}_{\ell}[t+1]}
\hat s_{j-\ell}\bigg|^2.
\end{align}
\end{lemma}
\begin{proof}[\quad\it{Proof:}]
See Appendix \ref{app:EM}.
\end{proof}

\begin{algorithm}[!t]
\caption{EM-based Channel Estimation}
\begin{algorithmic}[1]
\label{algo:EM_based}
\STATE{{\bf{Init:}}  $\boldsymbol{\hat{\beta}}$;}
\WHILE{the variation of the estimated channel is more than  $\varepsilon$ or the number of iterations does not exceed the maximum threshold}
\STATE{Update ${\hat a}_{\ell}$ and ${\hat\varphi}_{\ell}$, $\ell\in\mathcal{I}_{L}-1$, according to \eqref{LYat} and \eqref{LYphit1}, respectively;}
\STATE{Update ${\hat\sigma}^2$ according to \eqref{sigma};}
\ENDWHILE
\STATE{Update $\boldsymbol{\hat\beta}$, and then, feedback the updated $\boldsymbol{\hat\beta}$ to the Bayes equalizer module and the stop criteria module in the BERD receiver, respectively;}
\end{algorithmic}
\end{algorithm}

Given the Lemmas \ref{le:a_phi} and \ref{le:sigma}, the EM-based algorithm iterates between the E-step and M-step until the stop criteria are satisfied, i.e., $\Delta \boldsymbol{\hat{\beta}}={\big\|\boldsymbol{\hat{\beta}}[t+1]-\boldsymbol{\hat{\beta}}[t]\big\|^2}<\varepsilon$  or $t>t_{\text{max}}$.
Note that, in our blind channel and noise power estimation  problem, the E-step is actually used to determine the complete data ${\hat z}_{\ell,j}[t]$, $\ell\in\mathcal{I}_L-1$, $j\in\mathcal{I}_N$, which is given by \eqref{LYcomplete} in Appendix \ref{app:EM}.
The M-step can further update the estimates of the channel gain ${\hat a}_{\ell}[t+1]$ in \eqref{LYat} and the channel phase ${\hat \varphi}_{\ell}[t+1]$ in \eqref{LYphit1}, $\ell\in\mathcal{I}_L-1$.
Then, the noise power ${\hat\sigma}^2[t+1]$ is updated from \eqref{sigma}.
The convergence of the proposed EM-based estimation algorithm is analyzed in Appendix \ref{app:convergence}.

Update $\boldsymbol{\hat\beta}$ when the iteration of the EM-based algorithm stops.
The blind channel and noise power estimator outputs the estimated $\boldsymbol{\hat\beta}$ as the input of the Bayes equalizer to update the modulated symbol $\boldsymbol{\hat s}$ as Section \ref{sec:DEC}.
Meanwhile, the updated $\boldsymbol{\hat\beta}$ also serves as the input of the stop criteria module in the BERD receiver.
The algorithm and the information flow of the proposed blind channel and noise power estimator are shown in Figure \ref{fig:EM_module}.
Moreover, the proposed EM-based channel estimation algorithm is summarized in Algorithm \ref{algo:EM_based}.

\subsection{Initial of the EM-based Scheme}
\label{sec:initial}
The result achieved by the EM algorithm highly depends on the initial.
With a poor initial, the EM algorithm may converge to a local optimal solution far away from the global optimal one.
Thus, we first determine the initial of the unknown parameter $\boldsymbol \beta$.
In the literature, there are various methods to initialize the EM algorithm, such as the random restart \cite{ref:random-start}, the coarse grid search over the parameter space \cite{ref:modut2}, and the simulated annealing \cite{ref:SA1}.
However, in our problem, we consider the multipath channel of $L$ paths with the unknown parameters including the channel gain $\boldsymbol {a}$, the channel phase $\boldsymbol {\varphi}$, and the noise power ${\sigma}^2$, which need to be initialized simultaneously.
Hence, the dimension of the initial values is $2L+1$.\footnote{The typical values of $L$ could be $4$ \cite{FC1}, $6$ \cite{ref:ZJW}, etc.}
It is improper to adopt the random restart, the coarse grid search, or the simulated annealing as the initialization method since the computational complexity is exponential w.r.t. $L$, which is extremely involved.

To facilitate the initialization with a mild complexity and a good estimation performance, we adopt two initialization methods.
First, we initialize $\boldsymbol\beta$ by the true value of it with some bias, which is widely adopted in the literature, such as \cite{TWC[26],ref:ZJW,ref:ZJW-ICC}.
The initial channel gain, the initial channel phase, and the initial noise power are assumed to be uniformly distributed in the regions $(0,a_{\ell}+\delta_{a}]$, $[\varphi_{\ell}-\delta_{\varphi},\varphi_{\ell}+\delta_{\varphi}]$, and $(0,\sigma^2+\delta_{\sigma}]$, respectively, where $\delta_{a}$, $\delta_{\varphi}$, and $\delta_{\sigma}$ are the maximum biases for the unknown parameters, respectively.
Second, we apply a modified fourth-order moment-based method \cite{FC2,ref:ZJW} to initialize the unknown multipath channel and adopt the coarse grid search to initialize the noise power.
From \cite{FC2,ref:ZJW}, the fourth-order moment of the received signal is defined as $m_4^{{r}}(\kappa_1,\kappa_2,\kappa_3,\kappa_4)=\frac{1}{N}{\sum_{j\in\mathcal I_N}{r}_{j+\kappa_1}{r}_{j+\kappa_2}{r}_{j+\kappa_3}{r}_{j+\kappa_4}}$.
The normalized multipath channel coefficient $\hat {h}_{\ell}$ of the $\ell$th path is estimated by
\begin{align} \label{eq:initial_h}
\hat {h}_{\ell}
=\frac{m_4^{{r}}(\kappa,\kappa,\kappa,\kappa)}{m_4^{{r}}(\kappa,\kappa,\kappa,\ell)},
\quad \ell\in\mathcal{I}_{L}-1.
\end{align}
With the loss of generality, the leading path with $\kappa=0$ is assumed to be the dominant path.
Then, the initial values of the channel gain $a_\ell$ and the channel phase $\varphi_\ell$ can be determined directly from \eqref{eq:initial_h}, i.e., ${\hat a}_\ell = |{\hat h}_\ell|$ and ${\hat\varphi}_\ell =\angle {\hat h}_\ell$, $\ell \in \mathcal I_L-1$.
Then, we apply the coarse grid search algorithm \cite{ref:modut2} to initialize the noise power $\sigma^2$.
The parameter space of the coarse grid search is set to $(0,\frac 1 N \sum_{j\in\mathcal I_N }|{r}_j|^2)$ with a search step size $\alpha$.
By evaluating the log-likelihood function in \eqref{Lbeta} with the initial channel state information and the noise power of each grid in the parameter space, the initial of the noise power can be determined from \eqref{LYbeta}.

\section{Soft-information Detector and Regenerator}
\label{sec:DEC}
\begin{figure}[!t]
\centering
\includegraphics[width=4.5in]{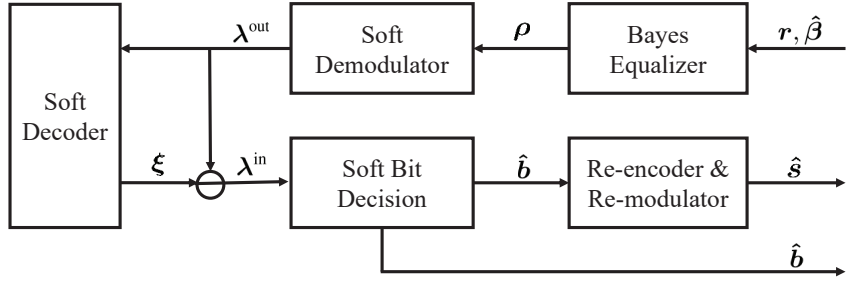}
\caption{The structure of the soft-information detector and regenerator.}
\label{fig:DEC_module}
\end{figure}
In this section, we first introduce the soft-information detector to determine the unknown information bits, which is the ultimate goal of this task. Then, the re-modulator and re-encoder module is introduced to obtain the modulated symbols $\boldsymbol{\sf s}$, which is required to estimate the unknown channel information in \eqref{LYat},  \eqref{LYphit1} and \eqref{sigma}.
The soft-information detector is composed of
the Bayes equalizer module, the module of the soft demodulator, the soft decoder, and the soft bit decision, and the re-encoder and re-modulator module, which are shown in Figure \ref{fig:DEC_module}.
In the following, we introduce the detection and the regeneration process.

First, we employ the Bayes equalizer to equalize the multipath channel.
Define ${\rho}_{m,j}$ as the posterior probability of the constellation point $\mu_m$ in $\mathcal S$ given the $j$th received symbol and the previous $L-1$ modulated symbols.
Then, ${\rho}_{m,j}$ is expressed as
\begin{figure}[!t]
\centering
\subfigure[]{
\label{fig:M-ary-QAM}
\includegraphics[width=2.8in]{{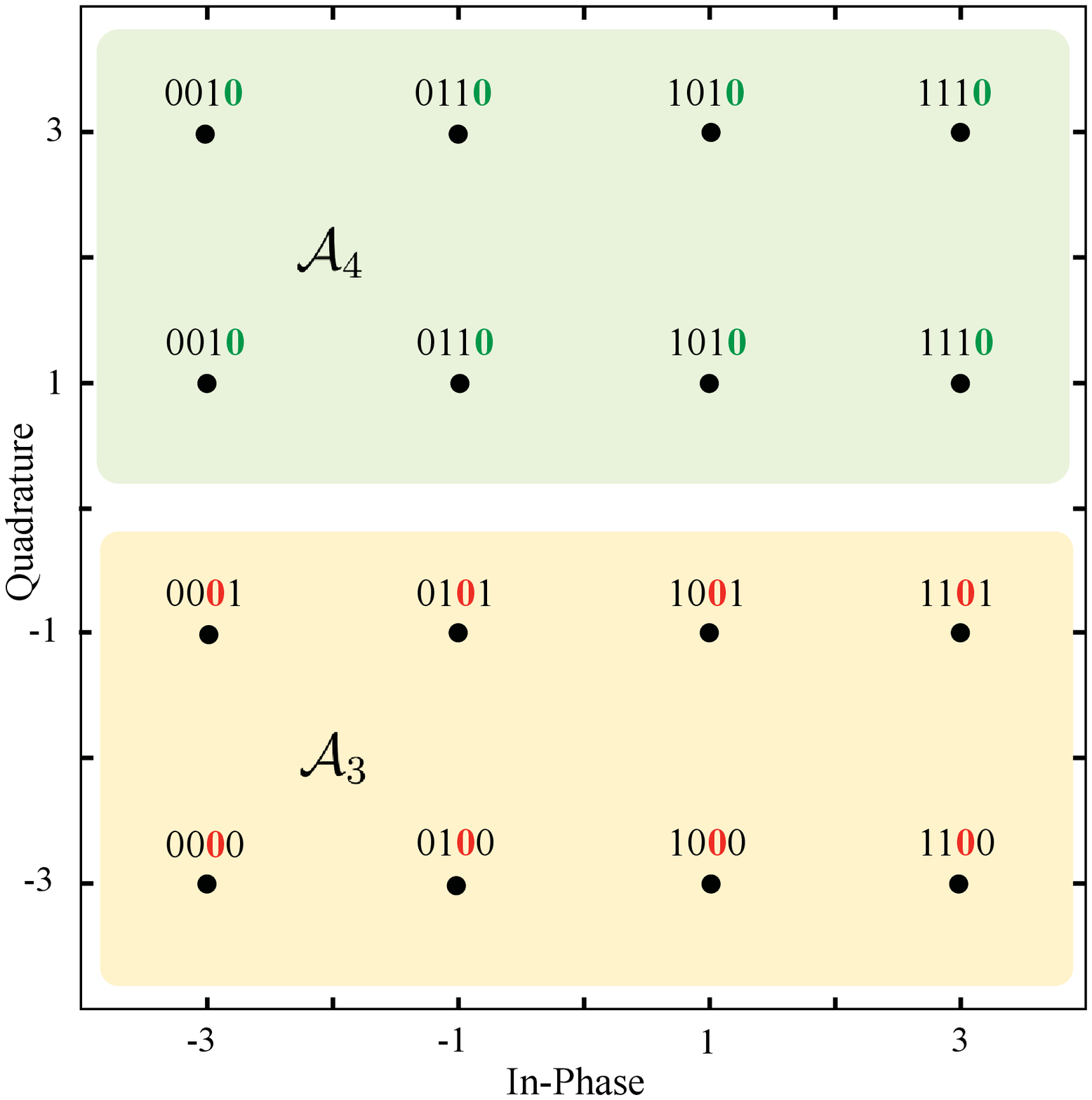}}}
\hspace{0.2in}
\subfigure[]{
\label{fig:M-ary-PSK}
\includegraphics[width=3.1in]{{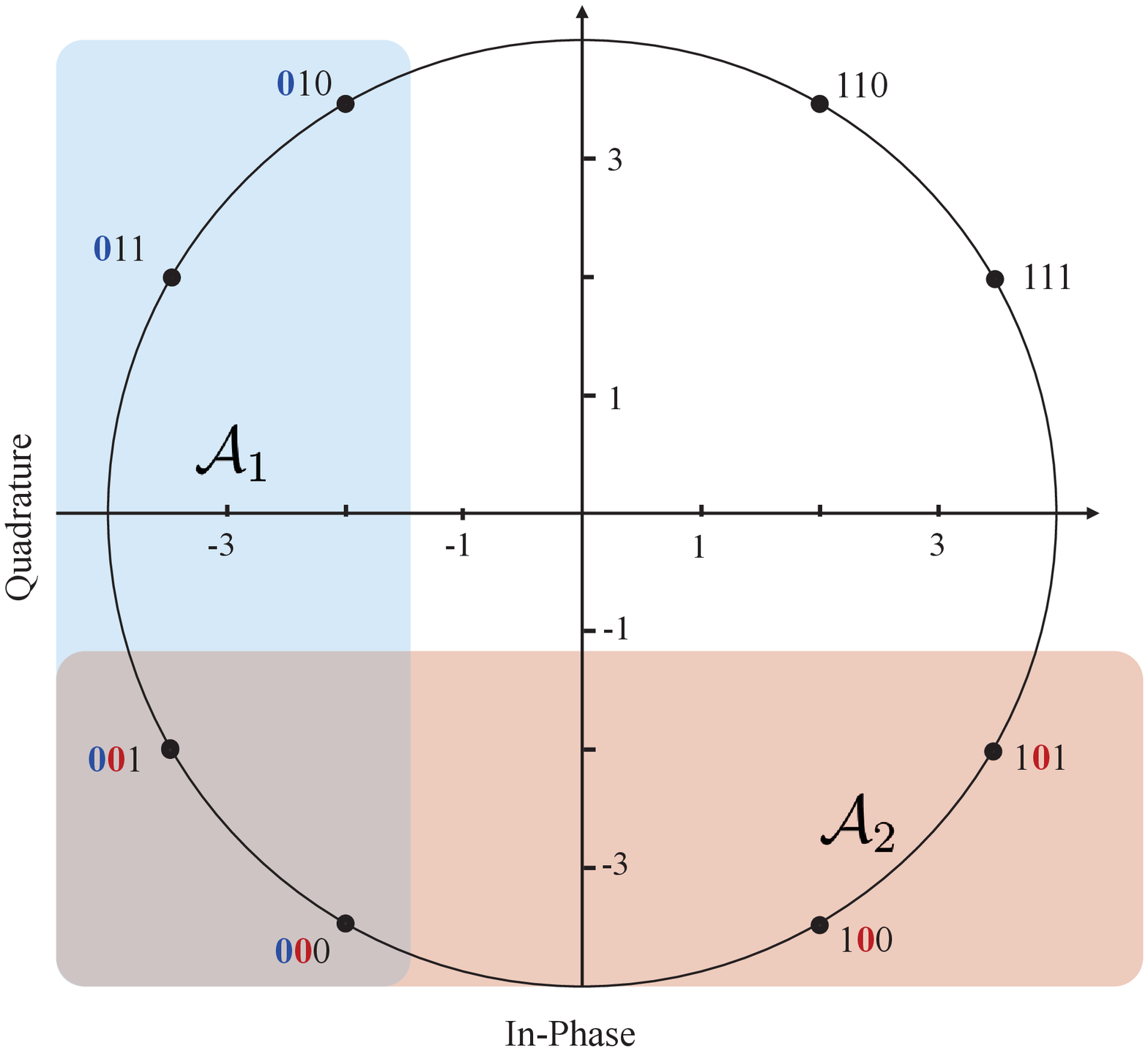}}}
\caption{The diagram of the constellation of $16$-QAM and $8$-PSK. (a) For $16$-QAM modulation, $\mathcal{A}_3$ and $\mathcal{A}_4$ represent the sets which contain all the constellation points with the third and the fourth bit being $0$. (b) For $8$-PSK modulation, $\mathcal{A}_1$ and $\mathcal{A}_2$ represent the sets which contain all the constellation points with the first and the second bit being $0$.}
\label{fig:M-ary}
\end{figure}
\begin{align}
\label{eq:rho_posteriori}
 {\rho}_{m,j}
&= {p}( {{\sf s}_{j}
= {{\mu}_m}| {r_j,\boldsymbol{\hat{s}}_{j-L+1}^{j-1};\boldsymbol{\hat{\beta}}} })\\
& = \frac{{{p}\big( {r_j|{{{\sf s}_{j}} = }{{\mu}_m},\boldsymbol{\hat{s}}_{j-L+1}^{j-1};\boldsymbol{\hat{\beta}}} \big)p\big( {{{\sf s}_{j}} = {{\mu}_m}|\boldsymbol{\hat{s}}_{j-L+1}^{j-1}; {\boldsymbol{\hat{\beta}}}} \big)}}{{\sum_{\mu_{m^{\prime}}\in\mathcal S}
{{p}\big( {{r}_j| {{{\sf s}_{j}} = } {{\mu}_{m^{\prime}}},\boldsymbol{\hat{s}}_{j-L+1}^{j-1};\boldsymbol{\hat{\beta}}} \big)p\big( {{{\sf s}_{j}}
= {{\mu}_{m^{\prime}}}|\boldsymbol{\hat{s}}_{j-L+1}^{j-1}; {\boldsymbol{\hat{\beta}}}} \big)} }} \label{eq:posti-and-modsym-bayes-rule}\\
& = \frac{{{p}\big( {r_j| {{\sf s}_{j} = }{{\mu}_m},\boldsymbol{\hat{s}}_{j-L+1}^{j-1};\boldsymbol{\hat{\beta}}} \big)}}{{\sum_{\mu_{m^{\prime}}\in\mathcal S}
{{p}\big( {r_j| {{\sf s}_{j} = } {{\mu}_{m^{\prime}}},\boldsymbol{\hat{s}}_{j-L+1}^{j-1};\boldsymbol{\hat{\beta}}} \big)} }},
\quad \mu_m\in\mathcal S, \quad j\in\mathcal{I}_{N}.
\label{eq:posti-and-modsym}
\end{align}
The equalization of \eqref{eq:posti-and-modsym-bayes-rule} follows from the Bayes rule, and
\eqref{eq:posti-and-modsym} is obtained by assuming that the transmitted symbols ${\sf s}_j$, $j\in\mathcal{I}_N$, are independent and each constellation point $\mu_m\in\mathcal S$ has an equal prior probability, i.e., $p\big( {\sf s}_{j} = {{\mu}_m}\big)=\frac{1}{|\mathcal{S}|}$.
Moreover, ${p}\big( {r_j| {{\sf s}_{j} = }{{\mu}_m},\boldsymbol{\hat{s}}_{j-L+1}^{j-1};\boldsymbol{\hat{\beta}}} \big)$ is given by
\begin{align}
\label{eq:p(rj|sj,s,beta)}
{p}\big( {r_j| {{\sf s}_{j} = }{{\mu}_m},\boldsymbol{\hat{s}}_{j-L+1}^{j-1};\boldsymbol{\hat{\beta}}} \big) \propto {\frac{1}{{\hat\sigma}} \exp \Bigg({-\frac{1}{{\hat\sigma^2}}{{{{\Big| {{r_{j}} - f({{s_{j}}=\mu_m})} \Big|}^2}}}} \Bigg)},\quad \mu_m\in\mathcal S, \quad j\in\mathcal{I}_{N}
\end{align}
and
\begin{equation}
\label{eq:f(sj)}
f( {{\sf s}_{j}}=\mu_m) = {{\hat a}_{0}}{e^{\imath{\hat\varphi_{0}}}}{\mu_{m}} + \sum_{\ell\in\mathcal{I}_{L-1}}{{{\hat a}_{\ell}}{e^{\imath{\hat\varphi_{\ell}}}}{\hat s}_{j-\ell}},\quad \mu_m\in\mathcal S, \quad j\in\mathcal{I}_{N}.
\end{equation}
Note that, to compute ${\rho}_{m,j}$ in \eqref{eq:posti-and-modsym}, $\boldsymbol{\hat{s}}_{j-L+1}^{j-1}=[{\hat s}_{j-L+1}, {\hat s}_{j-L+2},\ldots,{\hat s}_{j-1}]^{\mathrm T}$ are needed, which can be determined by
\begin{equation}
\label{LYsym}
{\hat s}_{j-\ell}= \sum_{\mu_m\in\mathcal S} {{{\mu}_{m}}{\rho}_{m,j-\ell}},
\quad j\in\mathcal{I}_{N}, \quad \ell\in\mathcal{I}_{L-1}.
\end{equation}

Hereafter, we adopt a soft demodulator to recover the symbols $\boldsymbol{\hat s}=[{\hat s}_1,{\hat s}_2,\ldots,{\hat s}_N]^{\mathrm T}$ by using the output  $\boldsymbol{\rho}=[{\rho}_{1,1},{\rho}_{2,1},\ldots,{\rho}_{\log|\mathcal S|,1},{\rho}_{1,2},\ldots,{\rho}_{\log|\mathcal S|,N}]^{\mathrm T}$ from the Bayes equalizer.
In general, a constellation point $\mu_m\in{\mathcal{S}}$ is corresponding to $\log|\mathcal S|$ coded bits.
We first define the coded bits as $\boldsymbol {\sf C} =[\boldsymbol {\sf c}_1,\boldsymbol {\sf c}_2,\ldots,\boldsymbol {\sf c}_N]\in\mathbb{Z}_2^{{\log |\mathcal S|}\times N}$, and $\boldsymbol {\sf c}_j=[{\sf c}_{j,1},{\sf c}_{j,2},\ldots,{\sf c}_{j,{\log |\mathcal S|}}]^{\mathrm T}\in\mathbb{Z}_2^{{\log |\mathcal S|}}$;
each $\boldsymbol {\sf c}_j$ maps to a constellation point in $\mathcal S$.
To describe the output of the soft demodulator explicitly, we define the constellation set $\mathcal{A}_g\subseteq\mathcal S$, $g\in\mathcal{I}_{\log |\mathcal S|}$, which contains all the constellation points with ${\sf c}_{j,g}=0$, $g\in\mathcal{I}_{\log |\mathcal S|}$.
Two examples are provided in Figure \ref{fig:M-ary}.
Then, the output posterior probability LLR ${\lambda}^{\text{out}}_{j,g}$ of the soft demodulator is denoted by
\begin{align}
{\lambda}^{\text{out}}_{j,g}& =\ln \frac{p({\sf c}_{j,g} =0|{r}_j,\boldsymbol{\hat{s}}_{j-L+1}^{j-1};\boldsymbol{\hat{\beta}})}
{p({\sf c}_{j,g}=1|{r}_j,\boldsymbol{\hat{s}}_{j-L+1}^{j-1};
\boldsymbol{\hat{\beta}})} \\
& =\ln \frac{\sum_{{\mu}_m\in\mathcal{A}_g}{\rho}_{m,j}}{ 1-\sum_{{\mu}_m\in\mathcal{A}_g} {\rho}_{m,j}},
\quad j\in\mathcal{I}_{N}, \quad g\in\mathcal{I}_{\log |\mathcal S|}.
\label{eq:lemeda_out}
\end{align}
After the soft demodulator, the output ${\boldsymbol\lambda}^{\text{out}}=[{\lambda}^{\text{out}}_{1,1},{\lambda}^{\text{out}}_{1,2},\ldots,{\lambda}^{\text{out}}_{1,\log|\mathcal S|},{\lambda}^{\text{out}}_{2,1},\ldots,{\lambda}^{\text{out}}_{N,\log|\mathcal S|}]^{\mathrm T}$ serves as the input of the soft decoder.
Define the information bits as $\boldsymbol{\sf b}=[{\sf b}_{1},{\sf b}_{2},\ldots,{\sf b}_{q}]^{\mathrm T}\in\mathbb{Z}_{2}^q$.
Then, the soft decoder outputs the posterior probability LLR  $\boldsymbol\xi=[\xi_{1},\xi_{2},\ldots,\xi_{q}]^{\mathrm T}$ as\footnote{In this section, we adopt the LDPC as an example and assume perfect synchronization which can be achieved by \cite{TB3}.
Thus, the $q$ information bits are encoded into $n$ coded bits, and then, are mapped to $N=\frac{n}{\log|\mathcal S|}$ modulated symbols.
Note that if $n$ cannot be divisible by $\log|\mathcal S|$, we can pad zero to guarantee that $\log |\mathcal S|$ divides $n$.
This operation is easy and trivial, and hence, we directly assume that $\log|\mathcal S|$ divides $n$.
In addition, the outputs $\xi_{l}$, $l\in\mathcal{I}_{q}$, of the soft decoder are determined by using the belief propagation algorithm in \cite{LLR1}.}
\begin{algorithm}[!t]
\caption{Soft-information Detection and Regeneration Algorithm}
\begin{algorithmic}[1]
\label{algo:re-sym}
\STATE{Compute $\boldsymbol\rho$ according to \eqref{eq:posti-and-modsym};}
\STATE{Compute ${\boldsymbol{\lambda}}^{\text{out}}$ according to \eqref{eq:lemeda_out};}
\STATE{Compute ${\boldsymbol\xi}$ according to \eqref{eq:BP-decode} and update ${\boldsymbol{\lambda}}^{\text{in}}$ by \eqref{eq:lemeda_in}, and compute the information bits ${\boldsymbol{\hat{b}}}$ by \eqref{eq:map_rule};}
\STATE{Regenerate $\boldsymbol{\hat {s}}$ in the re-encoder and re-modulator module, and output $\boldsymbol{\hat {b}}$ and $\boldsymbol{\hat {s}}$ to the stop criteria module in the BERD receiver.}
\end{algorithmic}
\end{algorithm}
\begin{align}
\label{eq:BP-decode}
{\xi}_l&=\ln \frac{p({\sf b}_l=0|{\Theta}^{\prime},\boldsymbol{\lambda}^{\text{out}})}
{p({\sf b}_l=1|{\Theta}^{\prime},\boldsymbol{\lambda}^{\text{out}})},
\quad l\in\mathcal{I}_{q}
\end{align}
where ${\Theta}^{\prime}$ is the parity-check relation in the hypothesis ${\theta^{\prime}}$.
Let $\boldsymbol{\tilde\lambda^{\text{out}}}=[{\tilde\lambda^{\text{out}}}_1,{\tilde\lambda^{\text{out}}}_2,\ldots,{\tilde\lambda^{\text{out}}}_q]^{\mathrm T}$, and the elements in $\boldsymbol{\tilde\lambda^{\text{out}}}$ equal to the first $q$ elements in $\boldsymbol{\lambda^{\text{out}}}$.
Furthermore, the updated extrinsic message ${\boldsymbol{\lambda}}^{\text{in}}=[{{\lambda}}^{\text{in}}_{1},{{\lambda}}^{\text{in}}_{2},\ldots,{{\lambda}}^{\text{in}}_q]^{\mathrm T}$ is given by
\begin{align}
\label{eq:lemeda_in}
{{\lambda}}^{\text{in}}_{l}= {\xi}_{l}-{\tilde\lambda}^{\text{out}}_{l},
\quad l\in\mathcal{I}_{q}.
\end{align}
After a hard decision of the soft bits, the information bits $\boldsymbol {\hat b}=[\hat b_1,\hat b_2,\ldots,\hat b_q]^{\mathrm T}$ can be obtained as
\begin{equation}
\hat{b}_{l}=\left\{
\begin{array}{rcl}
\label{eq:map_rule}
0 & & {\text{if ${{\lambda}_{l}^{\text{in}}}>0$}}\\
1 & & {\text{otherwise}, \quad l\in\mathcal I_q}.\\
\end{array} \right.
\end{equation}

Finally, the detected information bits $\boldsymbol {\hat b}$ input to the re-encoder and the re-modulator module to regenerate the modulated symbols $\boldsymbol{\hat s}$.
In addition, the outputs $\boldsymbol {\hat b}$ and $\boldsymbol{\hat s}$ serve as the inputs of the stop criteria module in the BERD receiver.
If the stop criteria are not satisfied, the updated $\boldsymbol{\hat s}$ is then fed to the blind channel and noise power estimator to update the estimate of the channel information $\boldsymbol{\hat\beta}$ in the next iteration.
The proposed soft-information detection and regeneration algorithm is summarized in Algorithm \ref{algo:re-sym}.\footnote{Note that the proposed soft-information detector and regenerator is not necessarily optimal or the fastest convergent since the extrinsic information is not eliminated in each operation during the detection and regeneration process.
The optimal or the fastest convergence detector and regenerator will be investigated in our future work.}

\section{Multistage Likelihood Decision Module}
\label{sec:multistage-decision}
We propose a multistage likelihood decision module to determine the information bits $\boldsymbol{\sf b}$, the adopted MCS $\theta=\{\eta, \zeta\}$, the unknown multipath channel and the noise power $\boldsymbol {\beta}$, as shown in Figure $\ref{fig:system_model}$(b).
To elaborate on the decision procedure explicitly, the hypothesis MCS candidate $\theta^{\prime}=\{\eta^{\prime},\zeta^{\prime}\}\in\mathcal{M}\times\mathcal{C}$ is used as the superscript of the decision metrics in this module.
In general, we first decide the modulation format, then, make the channel coding decision.
Finally, the information bits and the channel information are correspondingly determined.
In the following, we introduce the decision process.

\subsection{Modulation Decision}
In the modulation decision, the log-likelihood probability  $\mathcal F(\boldsymbol{\hat\beta}^{\eta^{\prime},\zeta^{\prime}})$, $\eta^{\prime}\in\mathcal{M}$, $\zeta^{\prime}\in\mathcal{C}$, in \eqref{Lbeta} is used as the modulation decision metric and the ML algorithm is adopted as the modulation classifier, which is denoted by
\begin{equation}
\label{eq:vote_modulation}
{{\tilde\eta}^{\zeta^{\prime}}}=\arg \max \limits_{\eta^{\prime}\in\mathcal{M}}{\mathcal F}(\boldsymbol {\hat\beta}^{\eta^{\prime},\zeta^{\prime}}),\quad \zeta^{\prime}\in\mathcal{C}.
\end{equation}
Hence, using ${{\tilde\eta}^{\zeta^{\prime}}}$, $\zeta^{\prime}\in\mathcal C$, we further determine the final decision of $\hat \eta$ by the majority vote in $\mathcal {M}$.
Hereafter, we need to recognize the channel coding scheme in $\mathcal C$ given the modulation $\hat\eta$.

\subsection{Channel Coding Decision}
In the channel coding decision, the average LLR of the syndrome APP is employed as the decision metric.
To derive this average LLR, we first provide the definition of the syndrome.
Given the channel coding
$\zeta^{\prime}\in\mathcal C$ and denote a non-zero vector $\boldsymbol{\pi}_{i}^{\hat\eta,{\zeta}^{\prime}}$ as the indices of the non-zero entries in the $i$th row of the parity-check matrix $\boldsymbol{H}^{\hat\eta,{\zeta}^{\prime}}$, i.e.,
${\boldsymbol \pi}_{i}^{\hat\eta,{\zeta}^{\prime}}=[\pi_{i}^{\hat\eta,{\zeta}^{\prime}}(1),\pi_{i}^{\hat\eta,{\zeta}^{\prime}}(2), \ldots,\pi_{i}^{\hat\eta,{\zeta}^{\prime}}(N_{i})]^{\mathrm{T}}$, $1\leqorig{\pi_{i}^{\hat\eta,{\zeta}^{\prime}}(1)}<\pi_{i}^{\hat\eta,{\zeta}^{\prime}}(2)<\cdots<\pi_{i}^{\hat\eta,{\zeta}^{\prime}}(N_{i})\leqorig{n}$,
where $N_{i}$ is the number of the non-zero elements in the $i$th row of $\boldsymbol{H}^{\hat\eta,{\zeta}^{\prime}}$.
Then, we have
\begin{equation}
\label{eq:syndrome}
{\tilde{\sf c}_{\pi_{i}^{\hat\eta,{\zeta}^{\prime}}(1)}} \oplus {\tilde{\sf c}_{\pi_{i}^{\hat\eta,{\zeta}^{\prime}}(2)}} \oplus  \cdots  \oplus {\tilde{\sf c}_{{\pi_{i}^{\hat\eta,{\zeta}^{\prime}}(N_{i})}}} = 0,\quad i\in\mathcal{I}_{n-q}.
\end{equation}
In general, if and only if $\hat\eta=\eta$ and ${\zeta}^{\prime}={\zeta}$, the relation \eqref{eq:syndrome} holds.
Furthermore, we define the LLR of the syndrome APP for the $i$th parity-check bit as $\gamma_i^{\hat\eta,\zeta^{\prime}}$, which is used to derive the average LLR metric.
To obtain $\gamma_i^{\hat\eta,\zeta^{\prime}}$, $i\in\mathcal{I}_{n-q}$, another lemma is provided as follows.

\begin{lemma}
\label{le:LLR-metric}
\em{Given the i.i.d. Bernoulli random variables ${\sf x}_j$, $j\in\mathcal I_N$, which takes the value $0$ with probability $p({\sf x}_j=0)$ and the value $1$ with probability $p({\sf x}_j=1)=1-p({\sf x}_j=0)$, the LLR metric in $\mathrm{GF}(2)$ is denoted by}
\begin{align}\label{eq:LLR-metric}
\mathcal{L}({\sf x}_1\oplus {\sf x}_2\oplus \ldots \oplus {\sf x}_N)
=2\tanh^{-1}\prod_{j\in\mathcal{I}_N}\tanh \frac{1}{2}{\mathcal{L}({\sf x}_j)}
\end{align}
where $\mathcal{L}({\sf x}_j)=\ln \frac{p({\sf x}_j=0)}{p({\sf x}_j=1)}$.
\end{lemma}
\makeatletter
\renewenvironment{proof}[1][\proofname]{\par%
\pushQED{\qed}%
\normalfont \topsep6\p@\@plus6\p@\relax%
\trivlist%
\item[\hskip\labelsep%
#1]\ignorespaces%
}{%
\popQED\endtrivlist\@endpefalse%
}
\makeatother
\begin{proof}[\quad\it{Proof:}]
See Appendix \ref{app:LLR-QAM}.
\end{proof}
Then, by using \eqref{eq:syndrome} and Lemma \ref{le:LLR-metric}, the LLR of the syndrome APP is specified by the following theorem.

\begin{Theorem}
\label{le:LLR-QAM}
\em{Given the modulation $\eta$ and a $(n, q)$ linear block code $\zeta$, the LLR of the syndrome APP for the $i$th parity-check bit is denoted by}
\begin{align}
\label{eq:LLR-QAM}
{\gamma}_{i}^{\eta,\zeta}
&= 2\tanh^{-1}\prod_{\tau\in\mathcal I_{N_i}}\tanh\frac{1}{2}{\psi}_{\pi_{i}^{\eta,\zeta}(\tau)},
\quad i\in\mathcal{I}_{n-q}
\end{align}
where $\psi_{\pi_{i}^{\eta,\zeta}(\tau)}\in\boldsymbol{\psi}$ is the posterior probability LLR of the $\pi_{i}^{\eta,\zeta}(\tau)$th coded bit in a codeword, and $\boldsymbol{\psi}=[\psi_{1},\psi_{2},\ldots,\psi_{n}]^{\mathrm T}$ is equal to ${\boldsymbol\lambda}^{\text{out}}$, which is obtained from \eqref{eq:lemeda_out}.
\end{Theorem}
\makeatletter
\renewenvironment{proof}[1][\proofname]{\par%
\pushQED{\qed}%
\normalfont \topsep6\p@\@plus6\p@\relax%
\trivlist%
\item[\hskip\labelsep%
#1]\ignorespaces%
}{%
\popQED\endtrivlist\@endpefalse%
}
\makeatother
\begin{proof}[\quad\it{Proof:}]
See Appendix \ref{app:LLR-QAM}.
\end{proof}

\begin{algorithm}[!t]
\caption{Multistage Likelihood Decision Algorithm}
\label{algo:decision}
\begin{algorithmic}[1]
\STATE{Decide the possible modulation candidates according to \eqref{eq:vote_modulation}. }
\IF{the modulation candidate having the majority vote in $\mathcal {M}$ is unique, i.e., $\hat \eta$}
\STATE{Compute ${\gamma}_{i}^{\hat\eta,\zeta^{\prime}}$, $i\in\mathcal{I}_{n-q}$, from Theorem \ref{le:LLR-QAM}.}
\STATE{Compute $\varGamma^{\hat\eta,\zeta^{\prime}}$ from \eqref{eq:vote_mod_decision}.}
\STATE{The channel coding $\hat\zeta$ is determined from \eqref{eq:vote_encoder_decision}.}
\ELSE
\STATE{Compute $\varGamma^{\theta^{\prime}}$ using \eqref{eq:vote_nomod_decision}.}
\STATE{The MCS $\hat\theta$ is determined from \eqref{eq:nomod_Taoiota}.}
\ENDIF
\STATE{$\boldsymbol{\hat{b}}^{\hat{\theta}}$ and $\boldsymbol{\hat{\beta}}^{\hat{\theta}}$ are determined correspondingly.}
\end{algorithmic}
\end{algorithm}

From Theorem \ref{le:LLR-QAM}, the average LLR of the syndrome APP $\varGamma^{\hat\eta,\zeta^{\prime}}$ for the channel coding decision is calculated by
\begin{equation}
\label{eq:vote_mod_decision}
\varGamma^{\hat\eta,\zeta^{\prime}} = \frac{1}{{n - q}}\sum\limits_{ i\in\mathcal{I}_{n-q}}{{\gamma}_{i}^{\hat\eta,\zeta^{\prime}}},
\quad {{\zeta}^{\prime}}\in\mathcal C.
\end{equation}
To further explain the relationship between the average LLR $\varGamma^{\hat\eta,\zeta^{\prime}}$ and the number of the parity-check bits, we have
\begin{align}
\label{eq:mod_Taoiota}
\varGamma^{\hat\eta,\zeta^{\prime}}(\iota) = \frac{1}{\iota}\sum_{i\in\mathcal I_\iota} {\gamma}_i^{\hat\eta,\zeta^{\prime}}, \quad\iota\in\mathcal{I}_{n-q},
\quad {{\zeta}^{\prime}}\in\mathcal C
\end{align}
and $\varGamma^{\hat\eta,\zeta^{\prime}}(\iota)$ represents the average LLR of the first $\iota$ parity-check bits, it is equal to $\varGamma^{\hat\eta,\zeta^{\prime}}$ in \eqref{eq:vote_mod_decision} if $\iota=n-q$.
Then, the decision of the channel coding is made as
\begin{equation}\label{eq:vote_encoder_decision}
{\hat\zeta} = \mathop {\arg \max \limits_{{{\zeta}^{\prime}} \in \mathcal{C} }} \varGamma ^{\hat\eta,{\zeta}^{\prime}}(\iota).
\end{equation}
Occasionally, the modulation decision step is incapable to determine a modulation format if the modulation candidate having the majority vote is not unique, we directly employ the average LLR $\varGamma ^{\theta^{\prime}}(\iota)$, $\iota\in\mathcal I_{n-q}$, as the decision metric for both the modulation classification and the channel coding recognition, which is given by
\begin{equation}
\label{eq:vote_nomod_decision}
\varGamma ^{{\theta ^{\prime}}}(\iota) = \frac{1}{\iota}\sum\limits_{i\in\mathcal{I}_{\iota}} {\gamma _{i}^{{\theta ^{\prime}}}}, \quad \iota\in\mathcal I_{n-q}, \quad \theta^{\prime}\in\mathcal M \times \mathcal C.
\end{equation}
Thus, the final decision of the adopted MCS $\hat\theta$ is made by
\begin{equation}
\label{eq:nomod_Taoiota}
{\hat\theta} = \mathop {\arg \max_{{\theta ^{\prime}} \in\mathcal{M}\times \mathcal{C} }} \varGamma ^{{\theta}^{\prime}}(\iota).
\end{equation}

\begin{algorithm}[!t]
\caption{The Proposed BERD Receiver}
\begin{algorithmic}[1]
\label{algo:joint}
\STATE{{\bf{Init:}} $\boldsymbol{\hat\beta}^{\theta^{\prime}}$ as Section \ref{sec:initial};}
\WHILE{the variation of the estimated channel is more than $\varepsilon$ or the number of iterations does not exceed the maximum threshold}
\STATE{Compute $\boldsymbol{\rho}^{\theta^{\prime}}$ in the Bayes equalizer module according to Algorithm \ref{algo:re-sym};}
\STATE{Detect $\boldsymbol{\hat b}^{\theta^{\prime}}$ in the soft demodulator and decoder module according to Algorithm \ref{algo:re-sym};}
\STATE{Regenerate $\boldsymbol {\hat s}^{\theta^{\prime}}$ in the re-encoder and re-modulator module according to Algorithm \ref{algo:re-sym};}
\STATE{Update $\boldsymbol{\hat\beta}^{\theta^{\prime}}$ in the blind channel and noise power estimator module as Algorithm \ref{algo:EM_based};}
\ENDWHILE
\STATE{The outputs including the detected bits $\boldsymbol{\hat b}^{\hat\theta}$, the MCS $\hat\theta$, and the estimated channel information $\boldsymbol{\hat\beta}^{\hat\theta}$ are determined in the multistage likelihood decision module according to Algorithm \ref{algo:decision};}
\end{algorithmic}
\end{algorithm}

The final decision of information bits is $\boldsymbol{\hat{b}}^{\hat{\theta}}$ correspondingly.
Furthermore, the multipath channel state information and the noise power are decided as $\boldsymbol{\hat{\beta}}^{\hat{\theta}}$.
The multistage likelihood decision algorithm is summarized in Algorithm \ref{algo:decision},
and the overall BERD receiver is summarized in Algorithm \ref{algo:joint}.

\section{BERD Approach for Multiple Receivers}
\label{sec:BERD-K-receive}
In this section, we extend the proposed BERD approach to the system with multiple receivers, which further enhances the performance of the data detection, the MCS recognition, and the channel estimation.
For multiple receivers, the BERD approach supports both the distributed and the cooperative manners.
The notation is summarized in Table \ref{table:symbol-list-multiple} in Appendix \ref{app:table} for the convenience of the readers.

\subsection{Multiple Receivers in Cooperative Manner}
\subsubsection{System Model}
We first extend the BERD approach to the system with multiple receivers in a cooperative manner.
Assume that the number of the receivers is $K$ and the received signal at $k$th receiver is $\boldsymbol{\sf r}_{k,:}=[{\sf r}_{k,1},{\sf r}_{k,1},\ldots,{\sf r}_{k,N}]^{\mathrm T}$, then, the $j$th received symbol at $k$th receiver ${\sf r}_{k,j}$ is given by
\begin{equation}\label{eq:system-Kreveiver}
{{\sf r}_{k,j}} = \sum\limits_{\ell\in\mathcal{I}_{L}-1} {{{a}_{k,\ell}}{e^{\imath{\varphi}_{k,\ell}}}{\sf s}_{j-\ell}+ {{\sf v} _{k,j}}},
\quad  k\in\mathcal{I}_K, \quad j\in\mathcal{I}_N
\end{equation}
where ${a}_{k,\ell}>0$ and ${\varphi}_{k,\ell}\in[0,2\pi)$ are the unknown channel gain and the unknown channel phase of the $\ell$th path at the $k$th receiver;
${\sf v}_{k,j}$ is the noise at $k$th receiver which follows a CSCG distribution i.e., ${\sf v}_{k,j}\sim\mathcal{CN}(0,{\sigma}_{k}^{2})$.
In particular, with the cooperation of multiple receivers, the task of the BERD approach is to detect the information bits $\boldsymbol {\sf b}$, recognize the MCS $\theta$, and estimate the unknown channel information  $\boldsymbol{B}=[\boldsymbol{\beta}_1,\boldsymbol{\beta}_2,\ldots,\boldsymbol{\beta}_K]$, where $\boldsymbol{\beta}_k=[{a}_{k,0},{a}_{k,1},$ $\ldots,{a}_{k,L-1}, {\varphi}_{k,0},{\varphi}_{k,1},\ldots,{\varphi}_{k,L-1},\sigma_k^2]^{\mathrm T}$.

\begin{figure}[tbp]
\centering
\subfigure[]
{
\begin{minipage}{0.9\linewidth}
\centering
\includegraphics[width=5.55in]{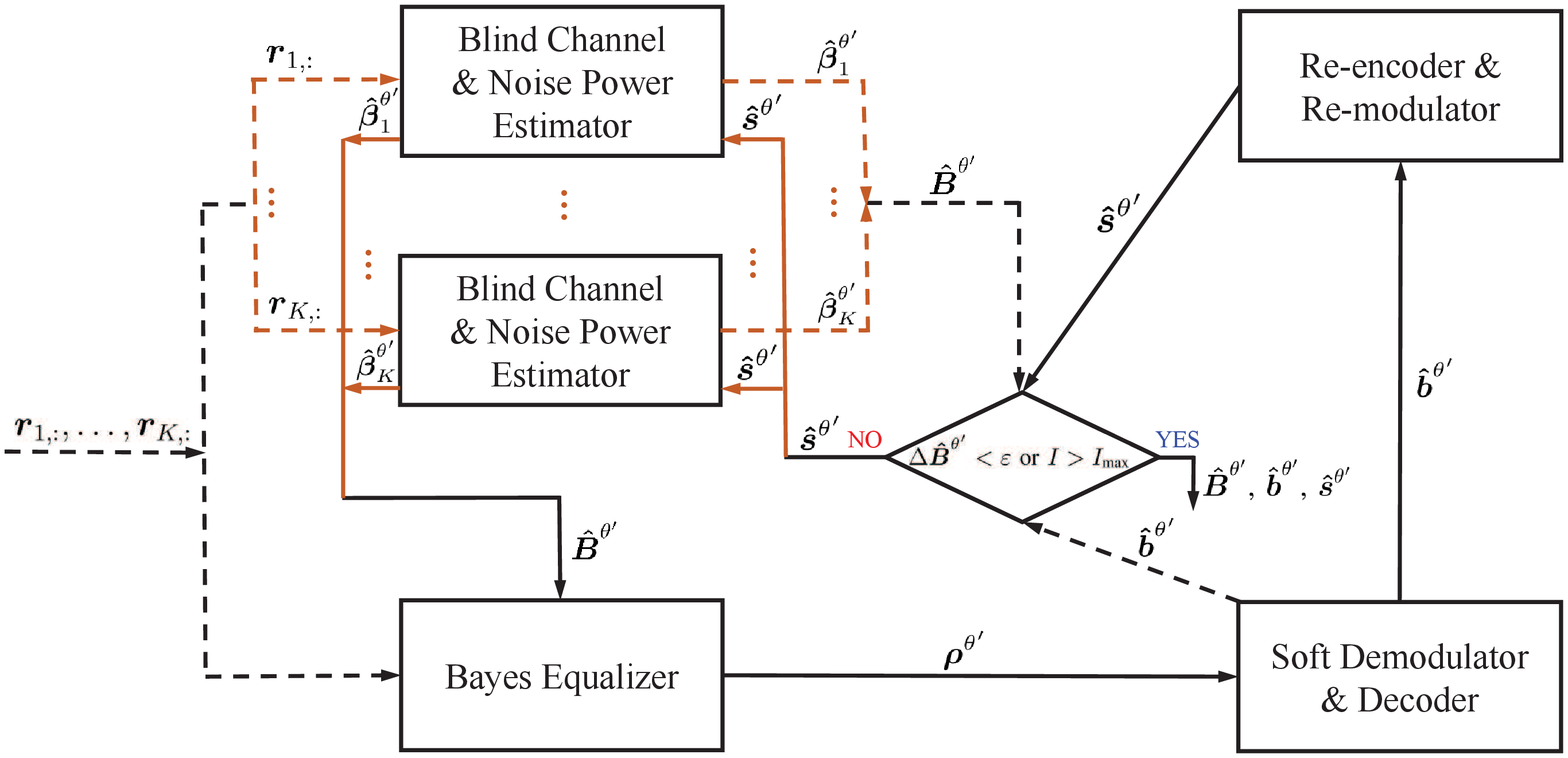}
\end{minipage}
}

\subfigure[]
{
\begin{minipage}{0.35\linewidth}
\includegraphics[width=2.7in]{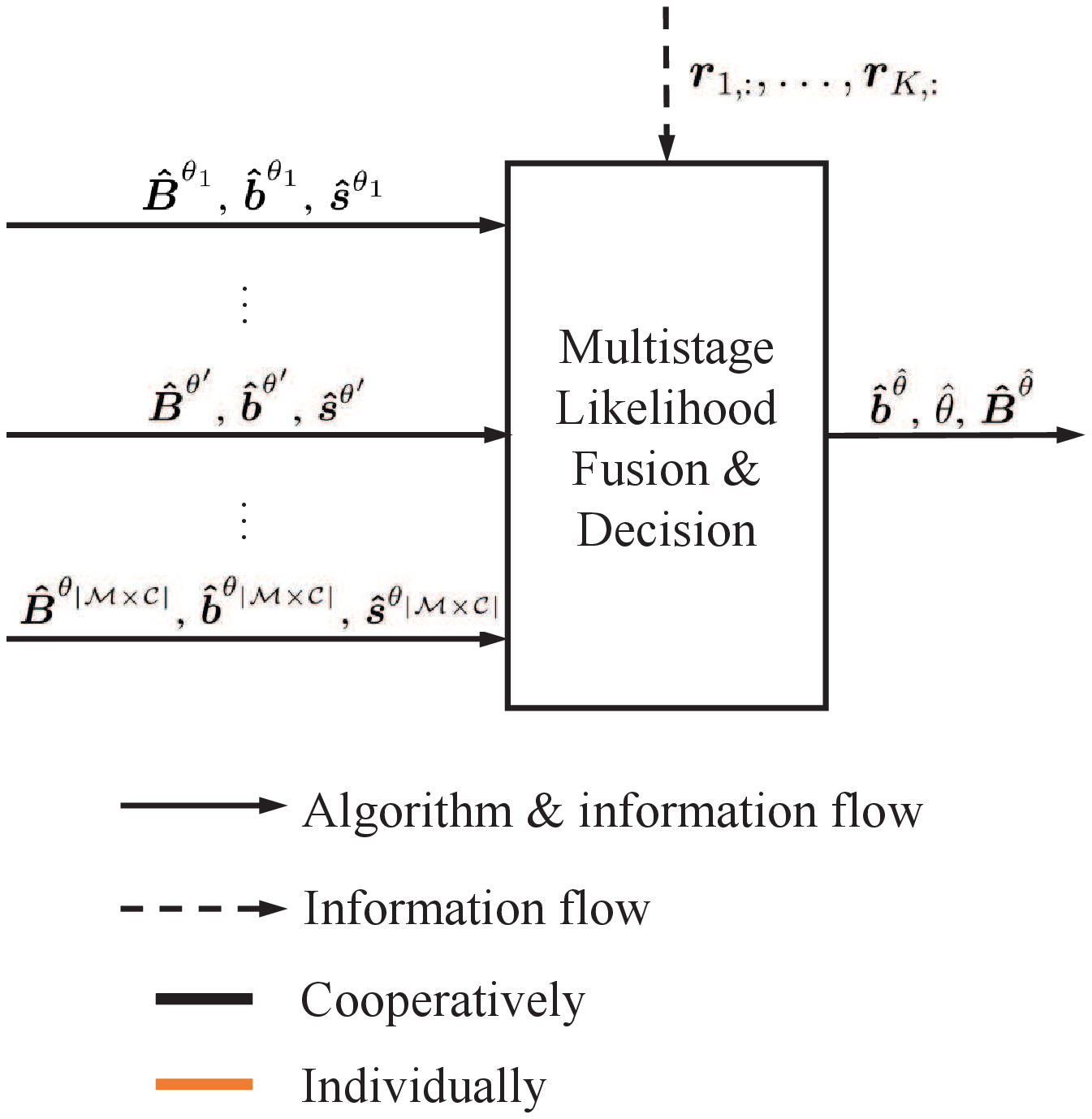}
\label{fig3}
\end{minipage}
}
\caption{The block diagram of the BERD approach for multiple receivers in a cooperative manner.
(a) The algorithm and information flow between the different  modules in the hypothesis MCS candidate $\theta^{\prime}$.
(b) The algorithm and information flow of the multistage likelihood fusion and decision module.
}
\label{fig:BERD_Kreceiver_cooperative}
\end{figure}

The essential procedure of the BERD approach for multiple receivers can be summarized as follows.
First, the multiple receivers individually estimate the multipath channel and the noise power $\boldsymbol{\hat\beta}_k$, $k\in\mathcal{I}_K$, which follows Algorithm \ref{algo:EM_based}.
Then, the soft-information detector and regenerator uses the estimation $\boldsymbol{\hat B}=[\boldsymbol{\hat\beta}_1,\boldsymbol{\hat\beta}_2,\ldots,\boldsymbol{\hat\beta}_K]$ from the multiple receivers to cooperatively detect the information bits $\boldsymbol {\hat b}$ and regenerate the modulated symbols $\boldsymbol{\hat s}$.
We revise some of the previous expressions for a single receiver, which are provided in Section \ref{sec:DEC_K_receiver}.
The BERD approach iterates between the blind channel and noise power estimator, and the soft-information detector and regenerator until the stop criteria are satisfied.
Finally, by utilizing $\boldsymbol{\hat B}$,  $\boldsymbol{\hat b}$, and $\boldsymbol{\hat s}$ determined in each hypothesis MCS $\theta^{\prime}\in\mathcal{M}\times\mathcal{C}$, the multistage likelihood fusion and decision module makes the final decision of the information bits, the MCS, and the channel information, which details are introduced in Section \ref{sec:decision-K-receiver}.
The extension of multiple receivers further enhances the performance of the detection, the recognition, and the estimation since the cooperative manner between the multiple receivers brings diversity gain.
The algorithm and the information flow of the proposed BERD approach for multiple receivers in cooperative manner are shown in Figure \ref{fig:BERD_Kreceiver_cooperative}.

\subsubsection{Soft-information Detector and Regenerator}
\label{sec:DEC_K_receiver}
The information bits $\boldsymbol{\sf b}$ and the modulated symbols $\boldsymbol {\sf s}$ are determined in a cooperative manner in the soft-information detector and regenerator.
The likelihood probability ${p}( {{\boldsymbol{r}_{:,j}}| {{{\sf s}_{j}} ={\mu}_m},\boldsymbol{\hat s}_{j-L+1}^{j-1};\boldsymbol{\hat \beta}})$ in \eqref{eq:p(rj|sj,s,beta)} is rewritten as
\begin{align}
\label{eq:likelihood_probability_Kreceiver}
& {p}( {{\boldsymbol{r}_{:,j}}| {{{\sf s}_{j}} ={\mu}_m},\boldsymbol{\hat s}_{j-L+1}^{j-1};\boldsymbol{\hat B}}) \notag \\
&\quad\quad \propto\frac{1}{\prod_{k\in\mathcal{I}_K}{\hat\sigma}_k} {\exp \Bigg( \sum_{k\in\mathcal{I}_{K}}
{-\frac{1}{{{\hat\sigma} _{k}^2}}
{{{{\big| {{{r}_{k,j}} - f_{k}({{\sf s}_{j}=\mu_m})} \big|}^2}}}} \Bigg)},
\quad \mu_m\in\mathcal S, \quad j\in\mathcal{I}_{N}
\end{align}
with $\boldsymbol{r}_{:,j}=[{r}_{1,j},{r}_{2,j},\ldots,{r}_{K,j}]^{\mathrm T}\in\mathbb{C}^{K}$ as the $j$th received symbols of $K$ receivers.
Let $\boldsymbol{\hat\beta}=\boldsymbol{\hat\beta}_k$, $f_{k}({{\sf s}_{j}}=\mu_m)$ is derived using \eqref{eq:f(sj)} at each receiver.
Then, the output posterior probability $\boldsymbol{\rho}$ of the Bayes equalizer is determined by plugging \eqref{eq:likelihood_probability_Kreceiver} into \eqref{eq:posti-and-modsym}.
After the Bayes equalizer, the methods to obtain the output of the soft demodulator $\boldsymbol{\lambda}^{\text{out}}$, the output of the soft decoder ${\boldsymbol\xi}$, the detected information bits $\boldsymbol{\hat b}$, and the regenerated modulated symbols $\boldsymbol{\hat s}$ are the same as that in Section \ref{sec:DEC}.

\subsubsection{Multistage Likelihood Fusion and Decision Module}
\label{sec:decision-K-receiver}
For the case of multiple receivers, the general idea of how to make the final decision is the same as the multistage likelihood decision module proposed in Section \ref{sec:multistage-decision}.
However, considering the cooperative manner of multiple receivers, we should modify some of the formulas in section \ref{sec:multistage-decision}.
In the modulation decision, the likelihood function ${\mathcal F}(\boldsymbol {\hat\beta})$ in \eqref{eq:vote_modulation} is computed in each hypothesis MCS $\theta^{\prime}=\{\eta^{\prime},\zeta^{\prime}\}\in\mathcal M \times \mathcal C$, which is rewritten as
\begin{align}
\mathcal F(\boldsymbol {\hat B})
= {\sum_{k\in\mathcal{I}_K}\sum_{j\in\mathcal{I}_N}-\frac{1}{{{\hat\sigma}_{k}^2}} {{\Big|{r}_{k,j}-\sum\limits_{\ell\in\mathcal{I}_{L}-1} {{\hat a}_{k,\ell}{e^{\imath{\hat \varphi}_{k,\ell}}}} {{\hat s}_{j-\ell}} \Big|^{2}}}}-N\ln\prod_{k\in\mathcal{I}_K}{\hat\sigma}_k.
\label{eq:Lbeta-K-receiver}
\end{align}
The decision of the channel coding $\hat\zeta$, the information data bits $\boldsymbol{\hat b}^{\hat\theta}$ and the channel information $\boldsymbol{\hat B}^{\hat\theta}$ performs as the methods in Section \ref{sec:multistage-decision}.

\subsection{Multiple Receivers in Distributed Manner}
\begin{figure}[!t]
\centering
\includegraphics[width=5.5in]{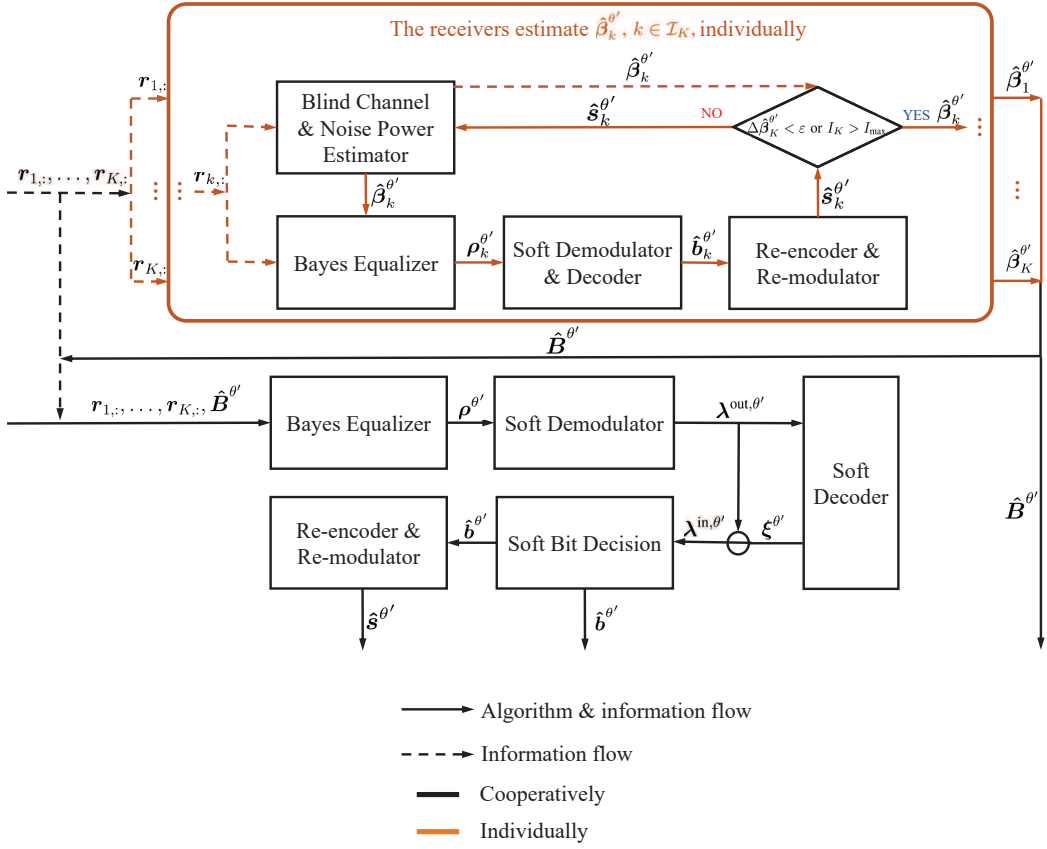}
\caption{The algorithm and information flow between the different  modules in the hypothesis MCS candidate $\theta^{\prime}$ of the BERD approach for multiple receivers in a distributed manner.
}
\label{fig:BERD_Kreceiver_distributed}
\end{figure}
Compared with a cooperative manner, the essential procedure of the distributed manner is summarized as follows.
First, the multiple receivers estimate the channel information, detect the information bits, and regenerate the modulated symbols individually instead of cooperation.
If the stop criteria are satisfied, multiple receivers  output the estimate $\boldsymbol{\hat B}^{\theta^{\prime}}$ to the soft-information detector and regenerator.
Then, the information bits $\boldsymbol{\hat b}^{\theta^{\prime}}$ and the modulated symbols $\boldsymbol{\hat s}^{\theta^{\prime}}$ are re-detected and re-generated before making the final decision in the multistage likelihood fusion and decision module.
Finally, utilizing $\boldsymbol{\hat B}^{\theta^{\prime}}$, $\boldsymbol{\hat b}^{\theta^{\prime}}$, and $\boldsymbol{\hat s}^{\theta^{\prime}}$ determined in each MCS $\theta^{\prime}\in\mathcal M\times\mathcal C$, the final decision of the MCS $\hat\theta$, the information data bits $\boldsymbol{\hat b}^{\hat\theta}$, and the channel information $\boldsymbol{\hat B}^{\hat\theta}$ performs as Section \ref{sec:decision-K-receiver}, shown in \ref{fig:BERD_Kreceiver_cooperative}(b).
The algorithm and the information flow between the different modules in the hypothesis MCS $\theta^{\prime}$ of the BERD approach for multiple receivers in a distributed manner are shown in Figure \ref{fig:BERD_Kreceiver_distributed}.

\section{Numerical Results}
\label{sec:simulation-results}
In this section, we provide various simulations to validate the proposed algorithm.
The number of channel paths is $L=6$ and the number of the receivers is $K=5$.
Without loss of generality, the leading coefficient of multipath channel is set to $1$, i.e., $h_{k,0}=a_{k,0}e^{\imath\varphi_{k,0}}=1$ and the remaining channel coefficients follow from the CSCG distribution with $\epsilon^{2}=0.1$ \cite{FC1,FC2,ref:ZJW}.

\noindent\emph{\textbf{Observation 1:}}
\noindent\emph{The average LLR metric $\varGamma^{\theta^{\prime}}(\iota)$ is much larger when the hypothetical MCS $\theta^{\prime}$ is accepted than that of being rejected, which indicates that the average LLR metric is effective for the recognition task. (c.f. Figure \ref{fig:LLR})}

\begin{figure}[!t]
\centering
\includegraphics[width=5in]{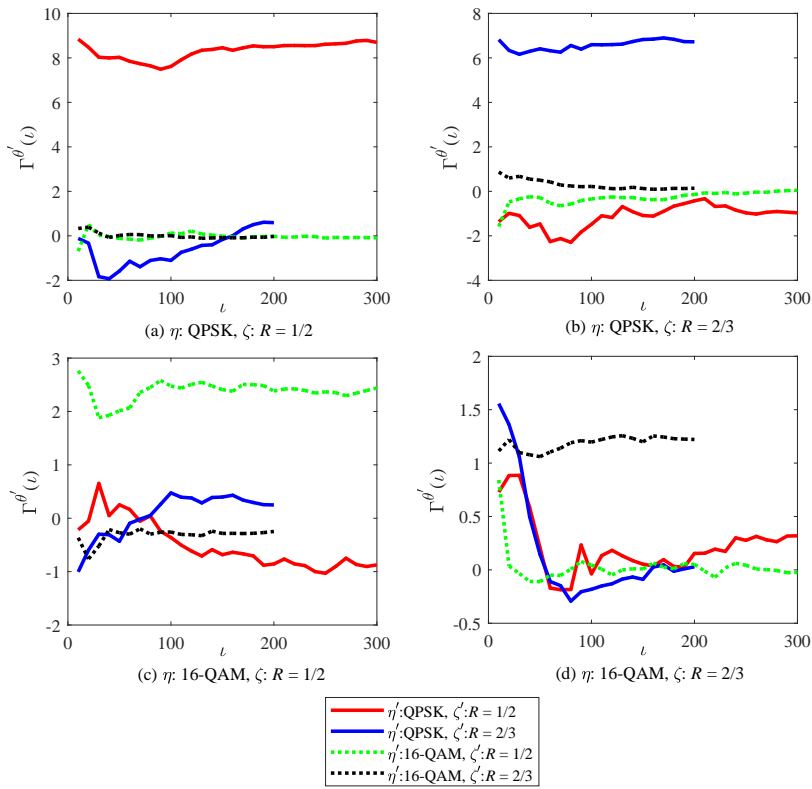}
\caption{The average LLR $\varGamma^{\theta^{\prime}}(\iota)$ is evaluated w.r.t. $\iota$, with curves parameterized by the different hypothesis MCS when $\mathsf{SNR}=2\,\mathrm{dB}$. The code length is $n=648$, and the number of the received symbols is $N=648$.
The adopted modulation $\eta$ and the channel coding $\zeta$ at the transmitter are (a) $\eta$: QPSK, $\zeta$: $R = \frac{1}{2}$; (b) $\eta$: QPSK, $\zeta$: $R = \frac{2}{3}$; (c) $\eta$: $16$-QAM, $\zeta$: $R = \frac{1}{2}$; (d) $\eta$: $16$-QAM, $\zeta$: $R = \frac{2}{3}$, respectively.}
\label{fig:LLR}
\end{figure}

In Figure \ref{fig:LLR}, we evaluate the characteristic of the average LLR of syndrome APP $\varGamma^{\theta^{\prime}}(\iota)$ for the first $\iota$ parity-check bits according to \eqref{eq:mod_Taoiota}.
We consider the modulation candidate set $\mathcal M$ as \{QPSK, $16$-QAM\}. %i.e., $|\mathcal{M}|=2$.
The encoder candidate set $\mathcal C$ contains the LDPC encoders with the code rate $\frac{1}{2}$ and $\frac{2}{3}$, and the code length is fixed at $n = 648$.
As Figure \ref{fig:LLR} shows, the adopted MCS $\theta$ at the transmitter are
\{$\eta$: QPSK, $\zeta$: $R=\frac 1 2$\}, \{$\eta$: QPSK, $\zeta$: $R=\frac 2 3$\}, \{$\eta$: $16$-QAM, $\zeta$: $R=\frac 1 2$\}, and \{$\eta$: $16$-QAM, $\zeta$: $R=\frac 2 3$\}, respectively.
In addition, we initialize the channel information $\boldsymbol\beta$ by the true value with some bias, and the maximum bias set is $(\delta_{a} = 0.1, \delta_{\varphi} = \frac{\pi}{20}$, $\delta_{\sigma} = 0.1)$ \cite{TWC[26]}.
From the simulation results, the average LLR $\varGamma^{\theta^{\prime}}(\iota)$ is always stay positive and it is larger when the hypothetical MCS $\theta$ are exactly the adopted $\theta$, i.e., $\zeta^{\prime} =\zeta$ and $\eta^{\prime} =\eta$.
For other hypothesis $\zeta^{\prime} \neq\zeta$ and/or $\eta^{\prime} \neq\eta$,
the average LLR is close to $0$ as the parity-check bits increases.

In the following, we illustrate the data detection performance, the recognition performance, and the channel information estimation performance of the proposed BERD receiver, where the BER, the correct recognition probability, and the MSE are adopted as the performance metric.
The modulation candidate set $\mathcal M$ is \{QPSK, $8$-PSK, $16$-QAM\}.
The LDPC codewords defined in IEEE 802.11ac standard are used in our simulations.
Three code lengths $n=648$, $1296$, and $1944$ are defined in this standard, and each code length corresponds to four different code rates $R=\frac{1}{2}$, $\frac{2}{3}$, $\frac{3}{4}$, and $\frac{5}{6}$.
We consider two initial schemes of $\boldsymbol\beta$.
The first one is the true value of it plus some bias, and the maximum bias set is $(\delta_{a} = 0.1, \delta_{\varphi} = \frac{\pi}{20}$, $\delta_{\sigma} = 0.1)$ \cite{TWC[26]}.
The second one is the fourth-order moment-based initial scheme introduced in Section \ref{sec:initial}, and the search step size $\alpha=0.1$.
In addition, we set $I_{\text{max}}=30$, $t_{\text{max}}=30$, and $\varepsilon=10^{-3}$.

\begin{figure}
\centering
\subfigure[ ]
{\begin{minipage}[b]{0.7\textwidth}
\includegraphics[width=1\textwidth]{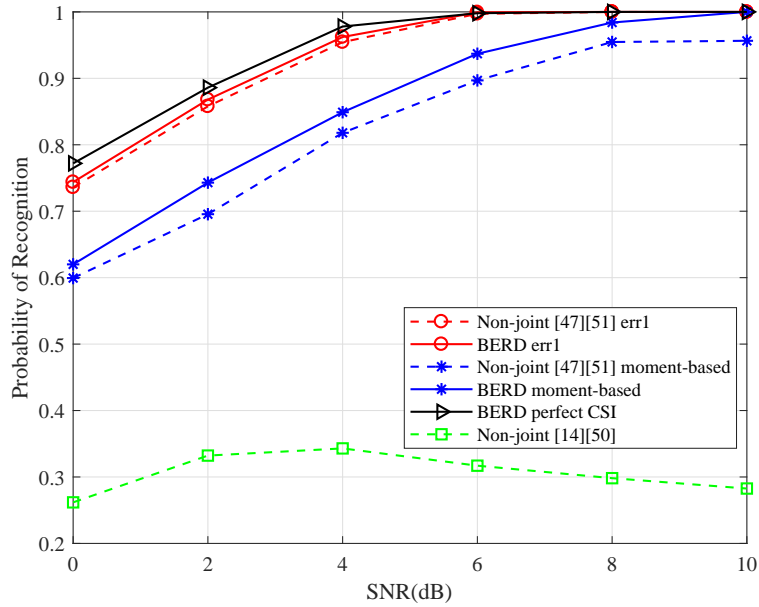}
\end{minipage}}
\subfigure[ ]
{\begin{minipage}[b]{0.7\textwidth}
\includegraphics[width=1\textwidth]{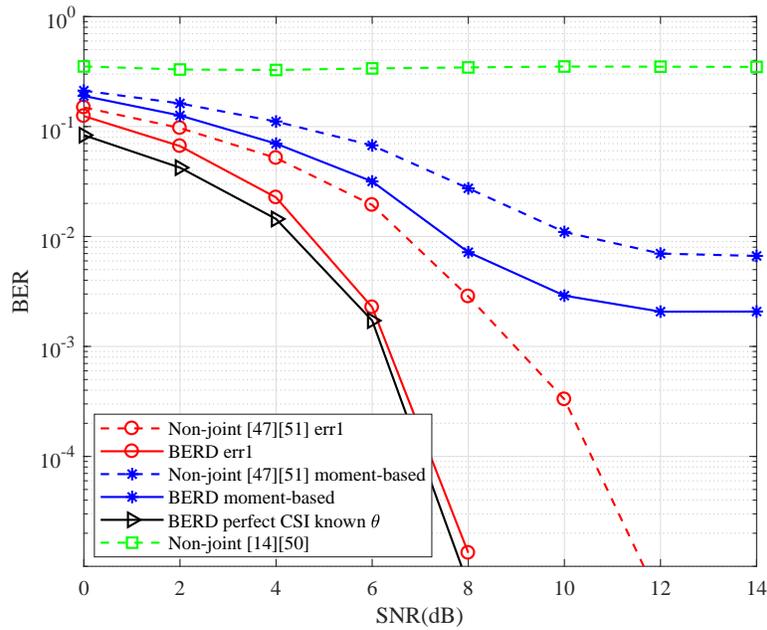}
\end{minipage}}
\caption{(a) The correct MCS recognition probability is evaluated w.r.t. SNR for the proposed BERD receiver compared to the existing schemes.
The benchmark is the BERD receiver with the perfect CSI.
(b) The BER is evaluated w.r.t. SNR for the proposed BERD receiver compared to the existing schemes.
The benchmark is the BERD receiver with the perfect CSI and the true MCS $\theta$.
The code length is $n=648$. The number of the received symbols is $N=648$.}
\label{fig:PCC-R}
\end{figure}

\begin{figure}
\centering
\subfigure[ ]
{\begin{minipage}[b]{0.7\textwidth}
\includegraphics[width=1\textwidth]{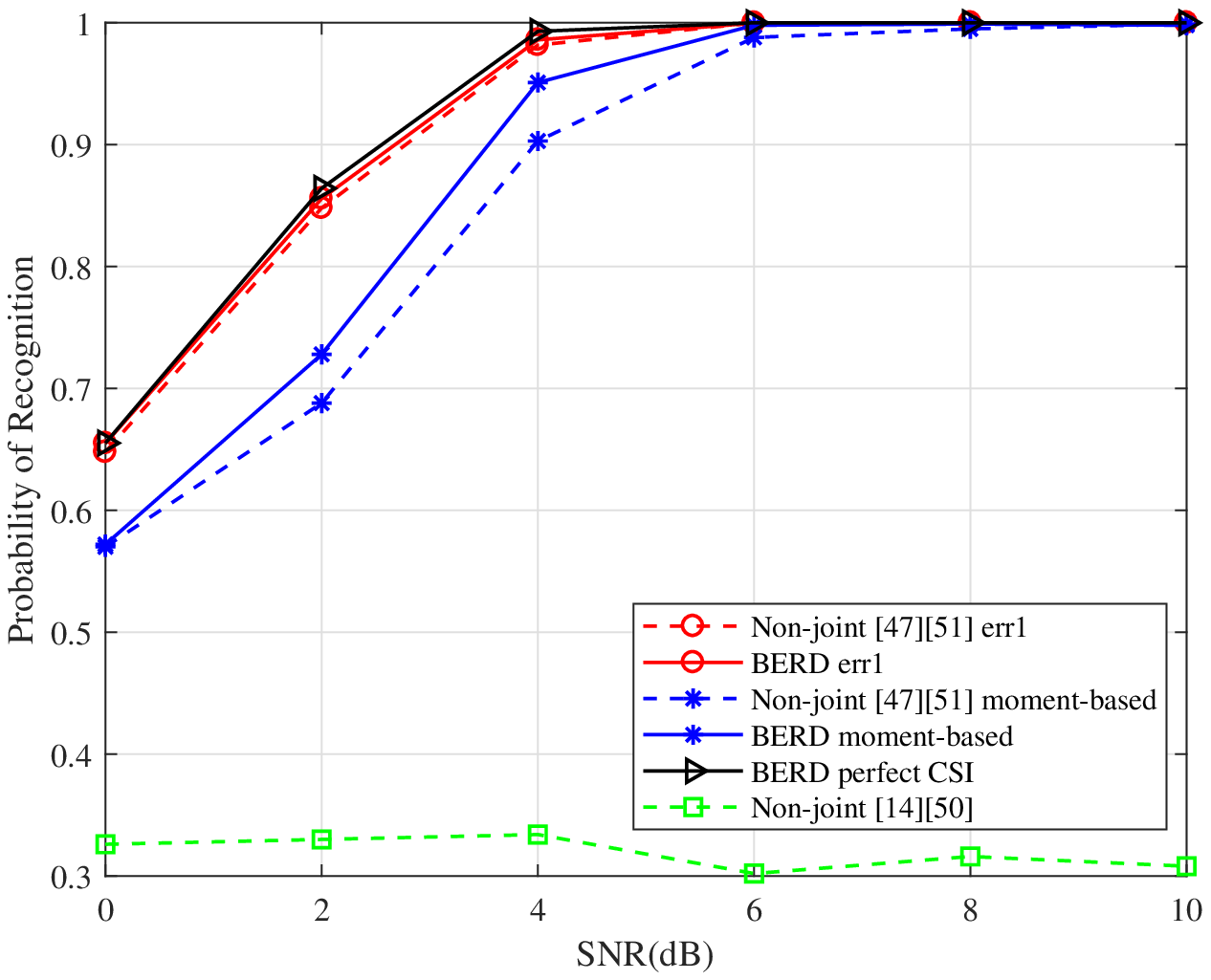}
\end{minipage}}
\subfigure[ ]
{\begin{minipage}[b]{0.7\textwidth}
\includegraphics[width=1\textwidth]{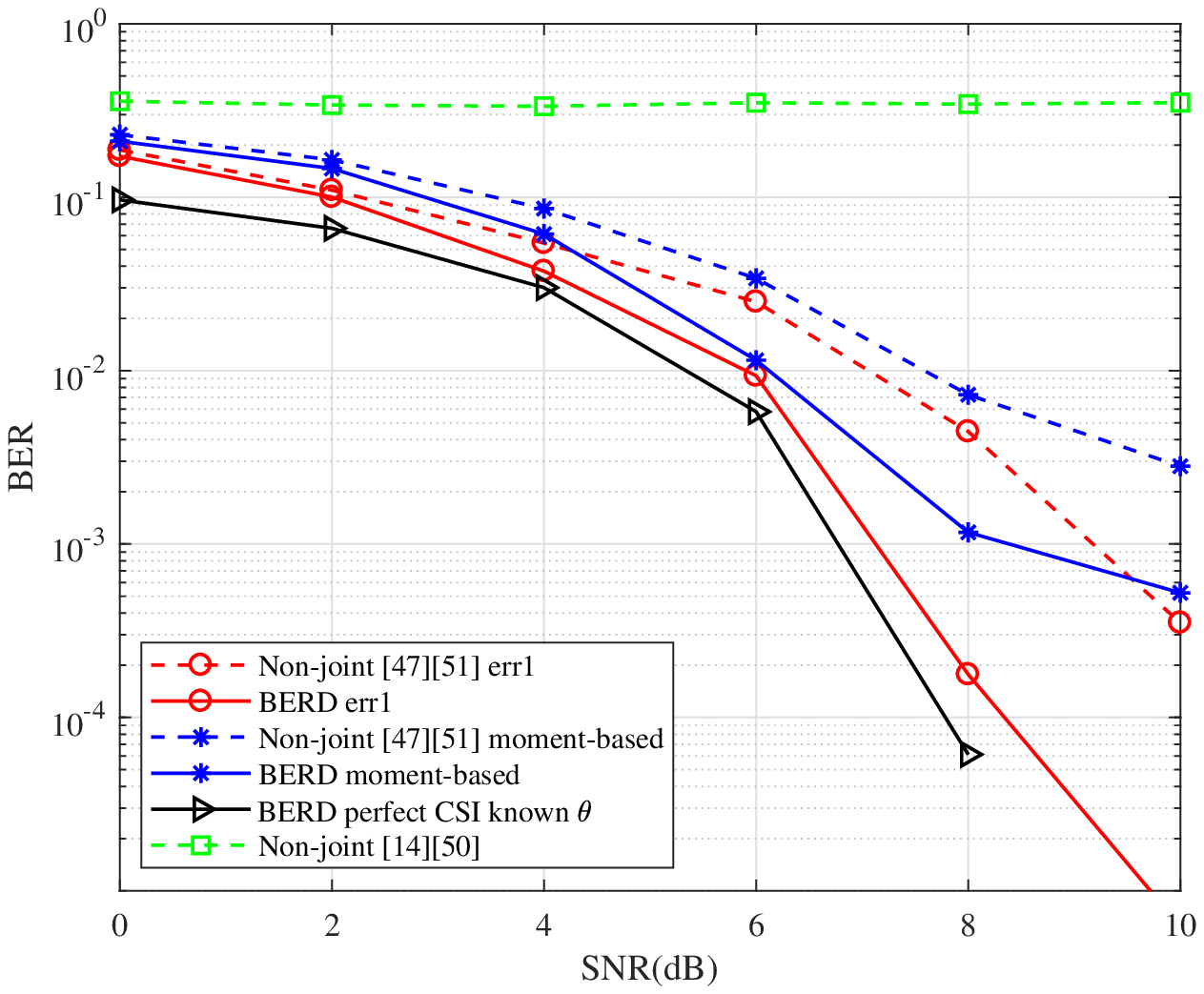}
\end{minipage}}
\caption{(a) The correct MCS recognition probability is evaluated w.r.t. SNR for the proposed BERD receiver compared to the existing schemes.
The benchmark is the BERD receiver with the perfect CSI.
(b) The BER is evaluated w.r.t. SNR for the proposed BERD receiver compared to the existing schemes.
The benchmark is the BERD receiver with the perfect CSI and the true MCS $\theta$.
The code rate is $R=\frac 5 6$. The number of the received symbols is $N=3888$.}
\label{fig:PCC-N}
\end{figure}

\noindent\emph{\textbf{Observation 2:}}
\noindent\emph{The proposed BERD receiver outperforms the existing schemes.
Moreover, with good initial, the MCS recognition performance is within $\mathit{0.3\,{dB}}$ as close to the one with the perfect CSI;
the loss of the BER is within $\mathit{0.5\,{dB}}$ at $\mathit{10^{-3}}$ compared with the one with the perfect CSI and the true MCS $\theta$. (c.f. Figures \ref{fig:PCC-R} and \ref{fig:PCC-N})
}

In Figures \ref{fig:PCC-R} and \ref{fig:PCC-N}, we evaluate the MCS recognition performance and the data detection performance of the proposed BERD receiver with the different initial schemes.
For MCS recognition, the BERD receiver with the perfect CSI serves as the benchmark.
Meanwhile, for data detection, the benchmark is the BERD receiver with the perfect CSI and the true MCS $\theta$.
We also compare our proposed BERD receiver with the existing schemes, which solve the overall problem by simply cascading the existing solution:
1) the first one was designed for the multipath scenario, which combines the approaches in \cite{ref:ZJW} and \cite{ref:LY};
2) the second one was designed for the single-path flat-fading scenario, which cascades the schemes in \cite{ref:modut2} and \cite{MULTI1}.
Note that the schemes in \cite{ref:ZJW} and \cite{ref:modut2} solve the modulation classification, the data detection, and the channel information estimation in the multipath channel and the single-path flat-fading channel, respectively; \cite{ref:LY} and \cite{MULTI1} tackle the channel coding identification, the data detection, and the channel information estimation in the multipath channel and the single-path flat-fading channel, respectively.
In Figure \ref{fig:PCC-R}, the channel coding candidate set $\mathcal{C}$ contains the encoders with different code rates, and the code length is $n=648$; while in Figure \ref{fig:PCC-N}, $\mathcal{C}$ contains the encoders with different code lengths, and the code rate is $R=\frac{5}{6}$.
In addition, the number of the received symbols at each receiver is $N=648$ and $N=3888$ in Figures \ref{fig:PCC-R} and \ref{fig:PCC-N}, respectively.
Note that, for the scheme \cite{ref:modut2} \cite{MULTI1}, we cannot initialize $\boldsymbol\beta$ by the true value of it with some bias, since $\boldsymbol\beta$ is treated as the single-path channel during the estimation and the dimension of the true value of $\boldsymbol\beta$ is not match to $\boldsymbol{\hat\beta}$.
Thus, only the fourth-order moment-based initial scheme is evaluated in this case.

From Figures \ref{fig:PCC-R}(a) and \ref{fig:PCC-N}(a), we can see that, for each initial scheme, the proposed BERD receiver achieves better MCS recognition performance than both the scheme \cite{ref:ZJW} \cite{ref:LY} and the scheme \cite{ref:modut2} \cite{MULTI1}.
Moreover, with good initial, the MCS recognition performance of the BERD receiver is within $0.3\,\mathrm{dB}$ as close to the benchmark.
From Figures \ref{fig:PCC-R}(b) and \ref{fig:PCC-N}(b), the data detection performance of the proposed BERD receiver outperforms the existing schemes with different initials.
Especially in the SNR region with the correct MCS recognition probability of over $90$\%, the gain of the BERD receiver in terms of the data detection is significant.
In addition, with the good initial scheme, the loss in the BER of the data detection is within $0.5\,\mathrm{dB}$ at $10^{-3}$ compared to the benchmark.

\begin{figure}
\centering
\subfigure[ ]
{\begin{minipage}[b]{0.7\textwidth}
\includegraphics[width=1\textwidth]{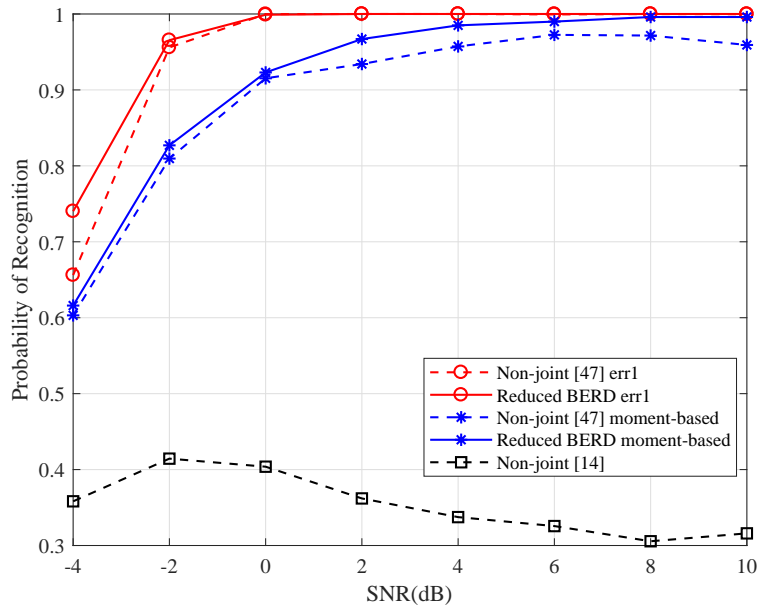}
\end{minipage}}
\subfigure[ ]
{\begin{minipage}[b]{0.7\textwidth}
\includegraphics[width=1\textwidth]{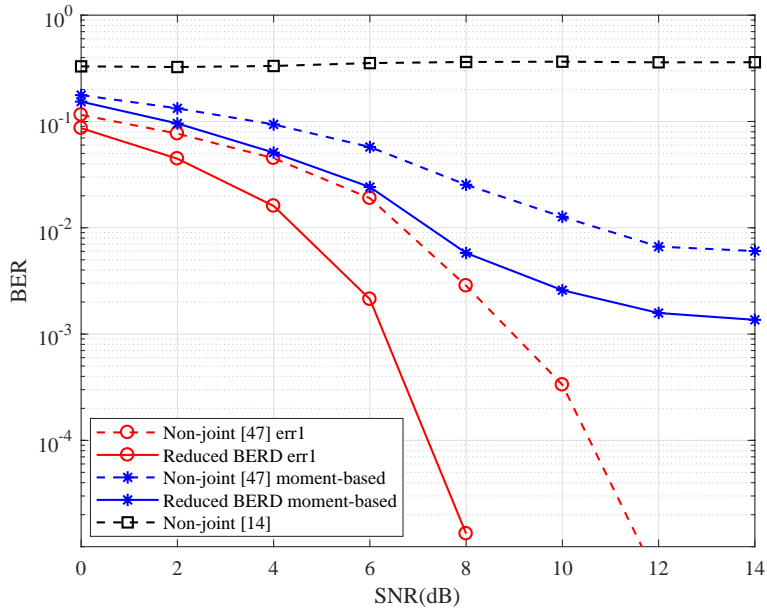}
\end{minipage}}
\caption{(a) The correct modulation classification probability is evaluated w.r.t. SNR for the reduced BERD receiver compared to the existing schemes.
(b) The BER is evaluated w.r.t. SNR for the reduced BERD receiver compared to the existing schemes.
The channel coding $\zeta$ is randomly selected from the encoder candidate set $\mathcal C$, which contains the encoders with different code rates, and the code length is $n=648$. The number of the received symbols is $N=648$.}
\label{fig:PCC-MC}
\end{figure}

\begin{figure}
\centering
\subfigure[ ]
{\begin{minipage}[b]{0.7\textwidth}
\includegraphics[width=1\textwidth]{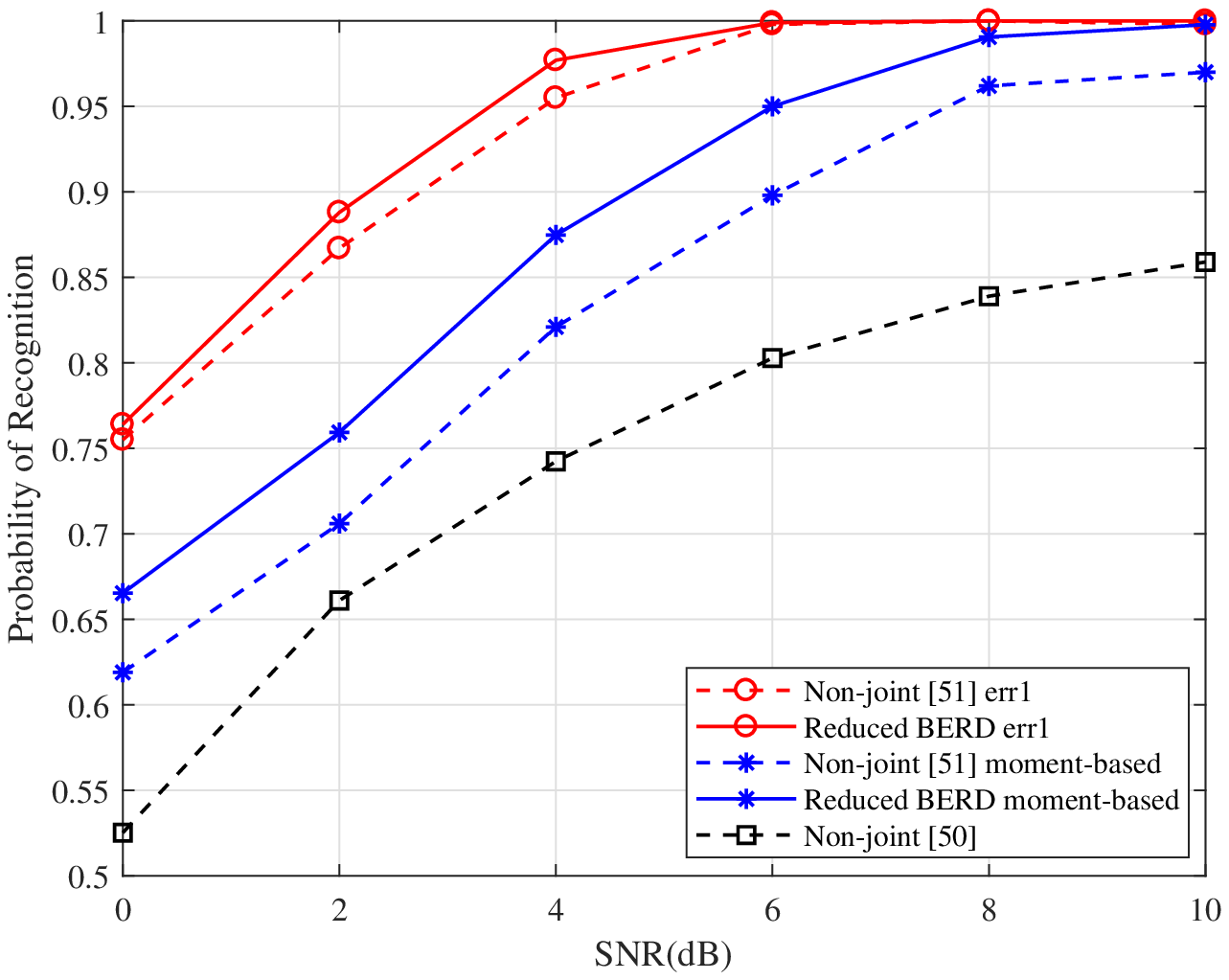}
\end{minipage}}
\subfigure[ ]
{\begin{minipage}[b]{0.7\textwidth}
\includegraphics[width=1\textwidth]{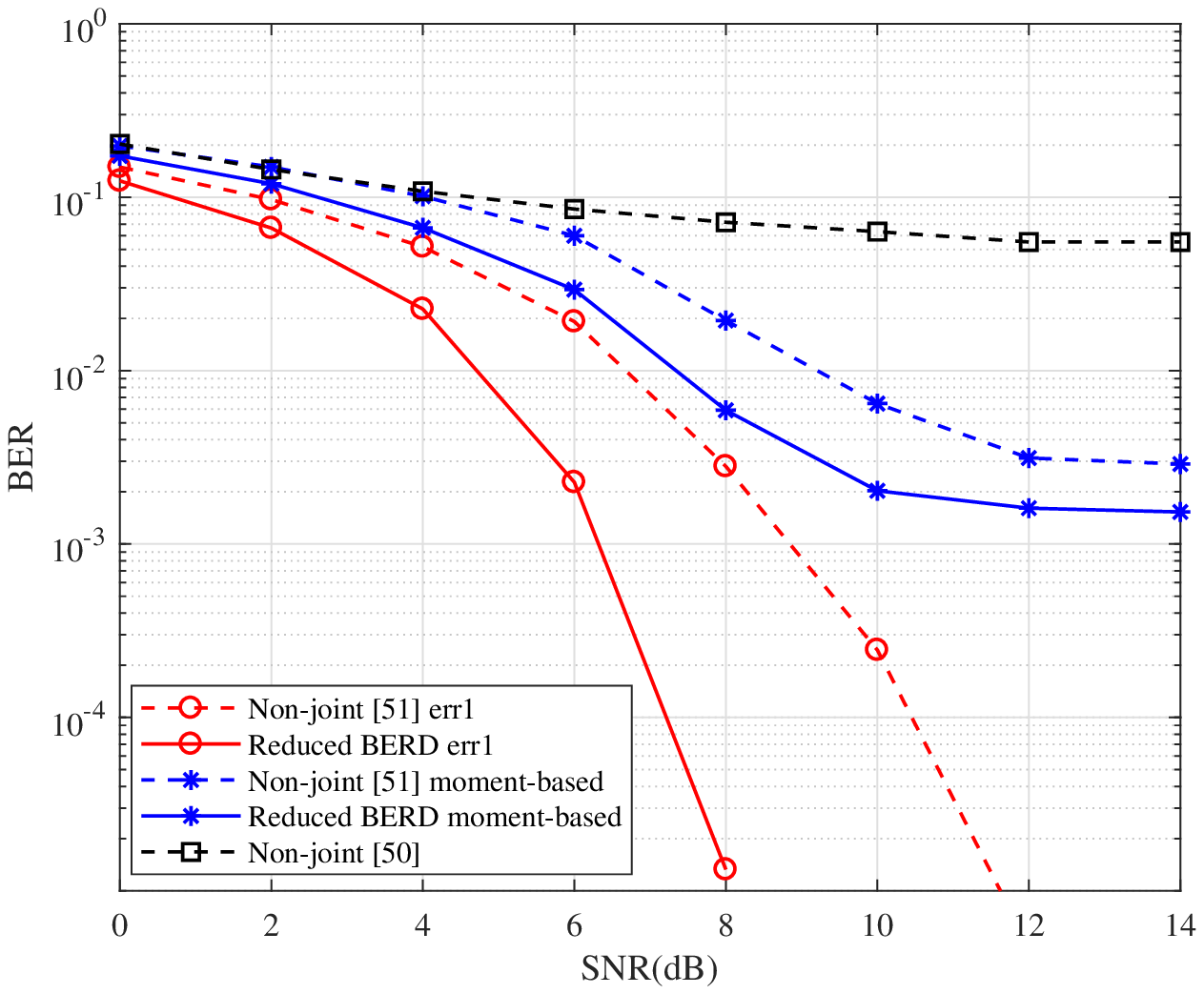}
\end{minipage}}
\caption{(a) The correct channel coding identification probability is evaluated w.r.t. SNR for the reduced BERD receiver compared to the existing schemes.
(b) The BER is evaluated w.r.t. SNR for the reduced BERD receiver compared to the existing schemes.
The modulation $\eta$ is randomly selected from the modulation candidate set $\mathcal M$.
The number of the received symbols is $N=648$.}
\label{fig:PCC-EC}
\end{figure}

In Figures \ref{fig:PCC-R} and \ref{fig:PCC-N}, compared to the scheme \cite{ref:ZJW} \cite{ref:LY}, the distinct merit of the BERD receiver lies in an iterative manner between the EM-based channel estimator and the soft-information detector.
To be specific, the soft-information detector corrects the errors from the Bayes equalizer and regenerates more reliable modulated symbols.
Then, the channel estimation is improved accordingly which further enhances the following data detection.
This iterative manner finally enhances the MCS recognition performance and decreases the BER.
Moreover, we can see that the scheme \cite{ref:modut2} \cite{MULTI1} cannot even achieve the acceptable performance.
This is because the single-path channel estimation method is improper for the multipath scenario, which further results in the low MCS recognition probability and the high BER.

\noindent\emph{\textbf{Observation 3:}}
\noindent\emph{The proposed BERD receiver can be applied to the reduced version of the original problem.
In addition, the reduced BERD receiver still outperforms the existing schemes. (c.f. Figures \ref{fig:PCC-MC} and \ref{fig:PCC-EC})}

In Figures \ref{fig:PCC-MC} and \ref{fig:PCC-EC}, we evaluate the correct recognition probability and BER performance for the reduced version of the original problem.
In Figure \ref{fig:PCC-MC}, we first demonstrate the reduced BERD receiver which contains the tasks of the data detection, the modulation classification, the multipath channel estimation, and the noise power estimation.
In this reduced case, the channel coding $\zeta$ is known at each receiver.
The adopted $\zeta$ is randomly selected from $\mathcal C$, which contains the encoders with different code rates, and the code length $n=648$.
Two schemes in \cite{ref:ZJW} and \cite{ref:modut2} are also evaluated for comparison.
Then, in Figure \ref{fig:PCC-EC}, we evaluate another reduced case, which involves the data detection, the channel coding identification, the multipath channel estimation, and the noise power estimation.
In this reduced BERD receiver, the modulation $\eta$ is known at each receiver, which is randomly selected from $\mathcal M$.
In addition, $\mathcal C$ is the same as Figure \ref{fig:PCC-MC}.
Two schemes in \cite{ref:LY} and \cite{MULTI1} are evaluated for comparison.

From Figures \ref{fig:PCC-MC}(a) and \ref{fig:PCC-EC}(a), we can see that the reduced BERD receiver achieves better modulation classification performance and channel coding identification performance than the existing schemes with different initials.
As Figures \ref{fig:PCC-MC}(b) and \ref{fig:PCC-EC}(b) show, the reduced BERD receiver outperforms the existing schemes in terms of the data detection.
Moreover, the performance gap between the reduced BERD receiver and the existing schemes becomes larger as the SNR increases.
The gain of the BER is attributed to the iterative manner between the EM-based channel estimator and the soft-information detector in the reduced BERD receiver.

\begin{figure}[!t]
\centering
\includegraphics[width=5in]{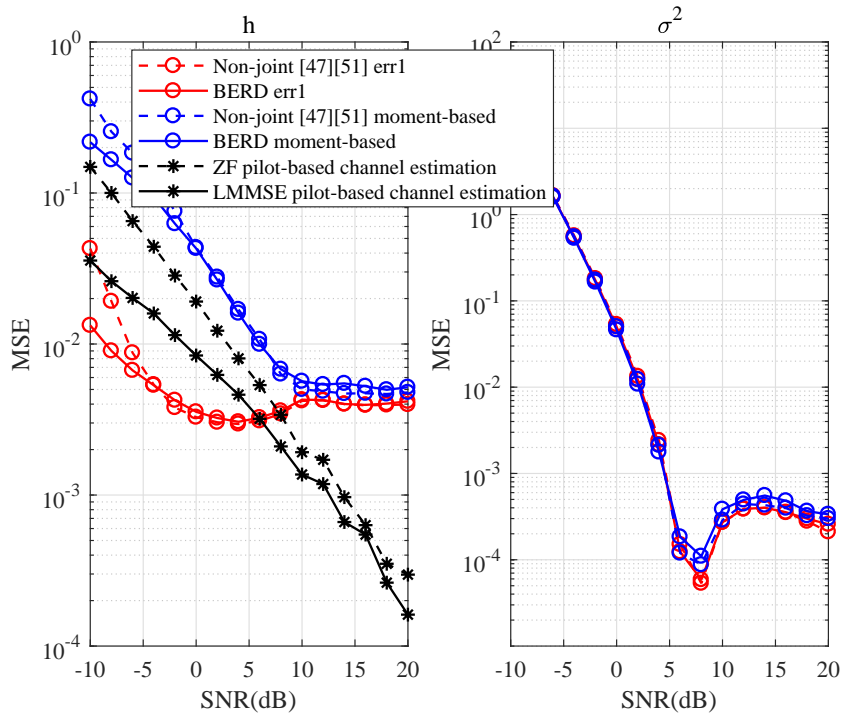}
\caption{(a) The MSE of the channel estimation is evaluated w.r.t. SNR for the proposed BERD receiver compared to the existing schemes.
The benchmarks are the ZF and LMMSE pilot-based channel estimation methods.
(b) The MSE of the noise power estimation is evaluated w.r.t. SNR for the proposed BERD receiver compared to the existing schemes.
The code length is $n=648$. The number of the received symbols is $N=648$.}
\label{fig:MSE}
\end{figure}

\noindent\emph{\textbf{Observation 4:}}
\noindent\emph{The BER performance gain of the proposed BERD receiver over the existing schemes is significant even the recognition performance is similar, which reveals the correct recognition does not guarantee the correct data detection. (c.f. Figures \ref{fig:PCC-R}-\ref{fig:PCC-EC})
}

From Figures \ref{fig:PCC-R}-\ref{fig:PCC-EC}, we observe that the proposed BERD receiver is able to provide slight performance gains in terms of the recognition task compared with the existing schemes; while the BERD receiver outperforms the existing approaches significantly in the BER performance, even when they have similar correct recognition probability.
This result reveals that the correct recognition does not guarantee the correct data detection.

\noindent\emph{\textbf{Observation 5:}}
\noindent\emph{The proposed BERD receiver achieves better channel estimation performance than the existing schemes in the low SNR region.
Moreover, with good initial, the MSE has $3.5\,\mathrm{dB}$ gain at $10^{-2}$ compared to the LMMSE pilot-based channel estimation method;
with worse initial, the loss in the MSE is within $3\,\mathrm{dB}$ at $10^{-2}$ compared to the ZF pilot-based channel estimation method. (c.f. Figure \ref{fig:MSE})}

In Figure \ref{fig:MSE}, we demonstrate the MSE performance of the multipath channel and noise power estimator in the BERD receiver.
The scheme proposed \cite{ref:ZJW} \cite{ref:LY} is evaluated for comparison.
In addition, we also provide the MSE performance of the ZF and LMMSE pilot-based channel estimation methods, which exploits all the transmitted data bits as pilots.
The channel coding candidate set $\mathcal C$ contains the encoders with different code rates, and the code length is $n=648$.

From Figure \ref{fig:MSE}, we can see that the channel estimation performance of the BERD receiver is better than the scheme \cite{ref:ZJW} \cite{ref:LY} in the low SNR region.
This is because the iterative manner between the EM-based channel estimator and the soft-information detector can bring MSE performance gain.
In addition, with good initial, the MSE of the channel estimation has $3.5\,\mathrm{dB}$ gain at $10^{-2}$ compared to the LMMSE pilot-based scheme;
with worse initial, the loss in the MSE of the channel estimation is within $2\,\mathrm{dB}$ gain at $10^{-2}$ compared to the ZF pilot-based scheme.
In addition, in Figure \ref{fig:MSE}(b), we can see that the MSE of the noise power estimation is nearly the same with different initials, which means the noise power estimation is not sensitive to the initial schemes.
Note that, with good initial, the MSE has a deterioration in the high SNR region. %$(\mathrm{SNR}\geqorig6\,\mathrm{dB})$.
The reason is that the power of the transmitted signal is set to $1$ in the simulations, then, the noise power decreases relatively as the SNR increases.
However, for the initialization of the noise power, the maximum bias is fixed to $\delta_{\sigma}=0.1$ during the estimation process, which means the relative bias of the initial is larger in the high SNR region.
On the other hand, as the SNR increases, the MLE problem in \eqref{LYbeta} contains more local optimal solutions.
Thus, even the initial is not far from the true value of $\boldsymbol\beta$, the estimation result $\boldsymbol{\hat\beta}$ is still more likely to be trapped in the local optimal solution, which lead to the worse MSE performance.

\section{Conclusion}
\label{sec:conclusion}
In this paper, a complete blind receiver approach named BERD was proposed, which can be applied in both the single receiver and the multiple receivers cases with the distributed manner or the cooperative one.
By iterating between the EM-based channel estimator and the soft-information detector, then, exploiting the likelihood fusion and decision module, the BERD receiver jointly solves the five tasks, including the blind multipath channel estimation, noise power estimation, modulation classification, channel coding identification, and data detection.
We show that the BERD receiver is extremely close to the benchmarks in terms of the MCS recognition and data detection,  and it outperforms the schemes which simply cascade the existing solution to each individual task.
Furthermore, the data detection performance of the reduced BERD receiver also outperforms the existing schemes, even when their recognition performances are similar.
In addition, with a good initial, the channel estimation performance of the proposed BERD receiver is close to the pilot-based methods in the low SNR region; while it floors in the high SNR region, which is not as good as the pilot-based ones.
The ramification of this paper is that an unknown signal can be recognized and decoded with quite little side information. It can be used to combat unknown interference in spectrum sharing or wiretap the information from an adversary, which finds many applications in both civilian and military scenarios.

\appendices
\section{Notations}
\label{app:table}
The notation and description of the BERD approach with a single receiver and multiple receivers are summarized in Tables \ref{table:symbol-list} and \ref{table:symbol-list-multiple}, respectively.

{\small
{
\setlength{\LTcapwidth}{\textwidth}
\begin{longtable}{p{2cm} p{13.8cm}}
\caption{Notation for the Case of a Single Receiver}
\label{table:symbol-list}\\
\toprule

\makebox[1.8cm][c]{Notation} & \makebox[13cm][c]{Description}\\

\midrule
\endfirsthead
\toprule
\endhead
\bottomrule
\endfoot
\endlastfoot

${a}_{\ell}$, ${\hat a}_{\ell}$ & The true and the estimated channel gain of the $\ell$th path, respectively \\

${{\varphi}}_{\ell}$, ${\hat{\varphi}}_{\ell}$ & The true and the estimated channel phase of the $\ell$th path, respectively \\

$\mathcal{A}_g^{\eta}$ & The set of all points which the $g$th bit is zero in the constellation $\eta$ \\

$\boldsymbol{\sf b}$, $\boldsymbol{\hat{b}}$ & The uncoded information bit sequence, the detected information bit sequence\\

$\boldsymbol{{\beta}}$, $\boldsymbol{\hat{\beta}}$ & The collection of the true and the estimated channel parameters including all channel gains, phases, and noise power, respectively \\

$\mathcal C$ & The set of candidate channel encoders \\

${\hat{c}_{j,g}}$ & The detected coded bit that maps to the $g$th bit in the $j$th modulated symbol \\

$\eta$, $\eta^{\prime}$ $\hat\eta$ & The unknown, the hypothesis, and the recognized modulation, respectively \\

${\mathcal F}(\boldsymbol\beta)$ & The log-likelihood function of $\boldsymbol\beta$  \\

$\boldsymbol{G}$, $\boldsymbol{H}$ & The generator matrix, the parity-check matrix\\

$\gamma_i^{\zeta^{\prime}}$ & The LLR of the syndrome APP of the $i$th parity-check bit in the hypothesis channel coding ${\zeta^{\prime}}$ \\

$\varGamma^{\zeta^{\prime}}(\iota)$ & The average LLR of the syndrome APP of the first $\iota$ parity-check bits in the hypothesis channel coding ${\zeta^{\prime}}$ \\

${\hat h}_\ell$ & The initial multipath channel of the $\ell$th path \\

$\varepsilon$, $I_{\text{max}}$, $t_{\text{max}}$ & The stop threshold, the maximum iterations of the BERD receiver and the blind channel estimator, respectively, \\

$L$ & The number of the paths of the wireless channel\\

${\hat\lambda}_{j,g}^{\text{out}}$ & The $g$th output LLR of the $j$th modulated symbol in the soft demodulator \\

${\hat\lambda}_{j,g}^{\text{in}}$ & The $g$th input LLR of $j$th modulated symbol in the soft bit decision module \\

$\mathcal M$ & The set of candidate modulation schemes\\

${\mu}_m$ & The $m$th constellation point in the constellation set\\

$n$ & The length of the codeword \\

$N$ & The number of the received symbols\\

$N_i$ & The number of the non-zero elements in the $i$th row of the parity-check matrix \\

$P$ & The sum of the power of the transmitted symbols\\

$q$ & The length of the uncoded information bit sequence \\

$R$ & The code rate of the codeword \\

${\sf r}_j$ & The $j$th received symbol \\

${\hat\rho}_{m,j}$ & The posterior probability of the $j$th modulated symbol maps to the $m$th constellation point \\

${\sf s}_j$ & The $j$th unknown modulated symbol\\

${\hat{s}}_j$ & The $j$th detected modulated symbol \\

$\boldsymbol{\hat s}_{j}^{j+L}$ & The detected symbol vector including the modulated symbols from ${\hat s}_j$ to ${\hat s}_{j+L}$ \\

$\mathcal S^{\eta}$ & The set of all constellation points in the constellation $\eta$\\

${\sf v}_j$ & The CSCG noise of the $j$th received symbol \\

${\sigma}^2$, ${\hat\sigma}^2$ & The true and the estimated noise power, respectively \\

${w}_{\ell}$ & The noise decomposition factor of the $\ell$th path\\

${{\sf z}_{\ell,j}}$ & The $j$th unknown complete data of the $\ell$th path\\

${{\sf z}_{\ell,j}}$, ${\hat{z}_{\ell,j}}$  & The $j$th determined complete data of the $\ell$th path\\

${\bar{\sf z}_{\ell,j}}$ & The signal which is obtained by taking the expectation of the noise component of ${{\sf z}_{\ell,j}}$ \\

$\theta$, $\theta^{\prime}$, $\hat\theta$ & The unknown, the hypothesis, and the recognized MCS, respectively \\

$\zeta$, $\zeta^{\prime}$ $\hat\zeta$ & The unknown, the hypothesis, and the recognized channel encoder, respectively \\

\bottomrule

\end{longtable}}}

{\small
{
\setlength{\LTcapwidth}{\textwidth}
\begin{longtable}{p{2cm} p{13.8cm}}
\caption{Notation for the Case of Multiple Receivers}
\label{table:symbol-list-multiple}\\
\toprule

\makebox[1.8cm][c]{Notation} & \makebox[13cm][c]{Description}\\
\midrule
\endfirsthead
\toprule
\endhead
\bottomrule
\endfoot
\endlastfoot
${a}_{k,l}, {\hat a}_{k,l}$ & The true and the estimated channel gain of the $\ell$th path at the $k$th receiver, respectively \\

${\varphi}_{k,l}$, ${\hat \varphi}_{k,l}$ & The true and the estimated channel phase of the $\ell$th path at the $k$th receiver, respectively \\

${\boldsymbol B}$, $\boldsymbol{\hat{B}}$ & The collection of the true and the estimated parameter including all channel gains, phases, and noise powers, respectively\\

${\boldsymbol\beta}_k$, $\boldsymbol{\hat{\beta}}_k$ & The collection of the true and the estimated parameters including all channel gains, phases, and noise power at the $k$th receiver, respectively\\

$K$ & The number of the multiple receivers\\

$\boldsymbol{\sf r}_{k,:}$ & The received signal at $k$th receiver \\

${\boldsymbol{\sf r}_{:,j}}$ & The $j$th received symbols of $K$ receivers \\

${\rv r}_{k,j}$ & The $j$th received symbol at the $k$th receiver\\

${\sigma}^2_{k}$, ${\hat\sigma}^2_{k}$ & The true and the estimated noise power at the $k$th receiver, respectively\\

${\sf v}_{k,j}$ & The noise of the $j$th received symbol at the $k$th receiver \\

${w}_{k,\ell}$ & The noise decomposition factor of $\ell$th path at the $k$th receiver \\

\bottomrule

\end{longtable}}}

\section{Proof of Lemma \ref{le:a_phi} and Lemma \ref{le:sigma}}
\label{app:EM}
We first proof the Lemma \ref{le:a_phi}, which states the closed-form expressions of the estimated channel gain and the estimated channel phase in \eqref{LYat} and \eqref{LYphit1}, respectively.
The details of how to design the EM-based channel estimator is provided in the following.

To deal with the MLE problem in \eqref{LYbeta} in a tractable way, the EM-based estimation algorithm is proposed to obtain the local optimal solution of the unknown $\boldsymbol{\beta}$.
In our problem, the E-step and M-step are formulated as
\begin{align}
\text{E-step}:& \quad J( {\boldsymbol{{\beta}}; {\boldsymbol{\hat\beta}[t]}} ) = {{\mathbb E}_{\boldsymbol{\sf Z}| {\boldsymbol{r},\boldsymbol{\hat s};\boldsymbol{\hat\beta}[t]} }} \big[ {\ln p( {\boldsymbol{\sf Z}|\boldsymbol{\hat s};\boldsymbol{{\beta}}})|{\boldsymbol{r}},\boldsymbol{\hat s};\boldsymbol{\hat\beta}[t]} \big] \label{eq:e-step}, \\
\text{M-step}:& \quad \boldsymbol{\hat\beta}[t+1] = \mathop {\arg \max \limits_{\boldsymbol{{\beta}}} }
 J( {\boldsymbol{{\beta}};{\boldsymbol{\hat\beta}[t]}}), \label{eq:m-step}
\end{align}
where $\boldsymbol{\sf Z}=[\boldsymbol {\sf z}_1,\boldsymbol {\sf z}_2,\ldots,\boldsymbol {\sf z}_N]
\in\mathcal{S}^{L\times N}$ is the complete data which cannot be obtained directly,
and $\boldsymbol {\sf z}_j=[{\sf z}_{0,j}, {\sf z}_{1,j},\ldots,{\sf z}_{L-1,j}]^{\mathrm{T}}\in\mathcal{S}^{L}$.
Additionally, $p(\boldsymbol{\sf Z}|\boldsymbol{\hat s};\boldsymbol{{\beta}})$ is the known density of $\boldsymbol{\sf Z}$.
Considering the multipath channel estimation problem in our BERD receiver, the received signal from the multipath channel is the summation of the signals from all the independent paths.
Hence, we choose the complete data as
\begin{align}
\label{eq:z_lj}
{\sf z}_{\ell,j} = {a}_{\ell}{e^{\imath{\varphi}_{\ell}}}{{\sf s}_{j-\ell}} + {{\sf v} _{\ell,j}}, \quad j\in\mathcal{I}_{N}, \quad \ell\in\mathcal{I}_{L}-1
\end{align}
where ${\sf v}_{\ell,j}$ is an i.i.d. CSCG distributed noise with the power
${\sigma}_{\ell}^{2}={w}_{\ell}{\sigma}^{2}$.
Define the noise decomposition factor as $\boldsymbol{w}=[{w}_0,{w}_1,\ldots,{w}_{L-1}]^{\mathrm T}$
and all the elements satisfy $\sum_{\ell\in\mathcal{I}_{L}-1} {w}_{\ell}=1$,\footnote{The choice of the noise decomposition factor $\boldsymbol{w}$ does not affect the estimation results, and the impact of $\boldsymbol{w}$ on the convergence of the proposed iterative algorithm is discussed in Appendix \ref{app:convergence}.} thus, the noise element ${\sf v}_{j}$ satisfies $\sum_{\ell\in\mathcal{I}_{L}-1}{\sf v}_{\ell,j}={\sf v}_{j}$.
Then, the relation between the received signal ${\sf r}_{j}$ and the complete data ${\sf z}_{\ell,j}$ is given by
\begin{align}\label{Fz}
{\sf r}_{j}& = \sum_{\ell\in\mathcal{I}_{L}-1}{{\sf z}_{\ell,j}} \quad j\in\mathcal{I}_{N}.
\end{align}
Let $\overline{\sf z}_{\ell,j}={a}_{\ell}e^{\imath{\varphi}_{\ell}} {\sf s}_{j-\ell}$.
Since the modulated symbols ${\hat{s}_{j}}$, $j\in\mathcal{I}_N$, have been determined by the soft-information detector and regenerator in Section \ref{sec:DEC},
$\overline{z}_{\ell,j}={a}_{\ell}e^{\imath{\varphi}_{\ell}} {\hat{s}_{j-\ell}}$ is the unknown deterministic signal.
Then, $\ln p\left( {\boldsymbol{\sf Z}|\boldsymbol{\hat{s}};\boldsymbol{{\beta}}} \right)$ in \eqref{eq:e-step} can be expressed as \cite{ref[TWC29]}
\begin{equation}\label{LYlnk1}
\ln p( {\boldsymbol{\sf Z}|\boldsymbol{\hat{s}}; \boldsymbol{{\beta}}}) = {C_1} - \sum\limits_{j\in\mathcal{I}_{N}} {\sum\limits_{\ell\in\mathcal{I}_{L}-1} {\frac{1}{{{\sigma} _{\ell}^2}}} } {\left| {{{\sf z}_{\ell,j}} - {{\bar {z}}_{\ell,j}}} \right|^2}
\end{equation}
where $C_{1}$ is a value that is independent of the blind channel estimation.
Then, given $\boldsymbol{r}$, $\boldsymbol{\hat s}$, and ${\boldsymbol{\hat\beta}[t]}$, the conditional expectation of \eqref{LYlnk1} is written as \cite{ref[TWC30]}
\begin{equation}\label{LYC2}
J( {\boldsymbol{{\beta}} ;{\boldsymbol{\hat\beta}[t]} } )
= {C_2} - \sum\limits_{j\in\mathcal{I}_{N}} {\sum\limits_{\ell\in\mathcal{I}_{L}-1} {\frac{1}{{\sigma} _{\ell}^2}} } {\left| {{\hat{z}_{\ell,j}[t]} - {a}_{\ell}e^{\imath{\varphi}_{\ell}}{{\hat s}_{j-\ell}}} \right|^2}
\end{equation}
where $C_{2}$ is another value independent of the blind channel estimation;
$\hat z_{\ell,j}$ is the conditional expectation of the $j$th complete data of the $\ell$th path.
Accordingly, the E-step in \eqref{eq:e-step} and the M-step in \eqref{eq:m-step} can be respectively simplified as

\noindent \quad E-step: Compute the complete data
\begin{equation}
\label{LYcomplete}
{\hat{z}_{\ell,j}[t]} = {\bar{z}_{\ell,j}[t]} + w_{\ell}\Bigg( {{r_{j}}
 - \sum\limits_{\ell\in\mathcal{I}_{L}-1} {\bar{z}_{\ell,j}[t]} } \Bigg),
  \quad \ell\in\mathcal{I}_{L}-1, \quad j\in\mathcal{I}_{N}
\end{equation}
\quad M-step: Estimate the channel information
\begin{equation}\label{LYbetamin}
\boldsymbol{\hat\beta}_{\ell}[t+1] = \mathop {\arg \min \limits_{\boldsymbol{{\beta}}_{\ell} }} \sum\limits_{j\in\mathcal{I}_{N}} {\frac{1}{{\sigma} _{\ell}^2}} {\left| {{\hat{z}_{\ell,j}[t]} - {a}_{\ell}e^{\imath{\varphi}_{\ell}}{{\hat s}_{j-\ell}}} \right|^2}, \quad \ell\in\mathcal{I}_{L}-1.
\end{equation}
It should be noted that by setting the derivative w.r.t. ${a}_{\ell}$ in \eqref{LYbetamin} to zero, we have the updated channel gain ${\hat a}_{\ell}[t+1]$ in \eqref{LYat}.
Since the second derivative of \eqref{LYbetamin} w.r.t. $a_\ell$ is a negative definite matrix, the equation \eqref{LYat} is the optimal estimate of ${a}_{\ell}$.
Then, substituting \eqref{LYat} into \eqref{LYbetamin}, we obtain the updated channel phase ${\hat\varphi}_{\ell}[t+1]$ in \eqref{LYphit1} with some straightforward operations.
Therefore, the proof of Lemma \ref{le:a_phi} is concluded.

Furthermore, we prove the estimated noise power in Lemma \ref{le:sigma} with Lemma \ref{le:a_phi}.
The noise-free signal is first determined by using the updated ${\hat a}_{\ell}[t+1]$ in \eqref{LYat}, the updated ${\hat \varphi}_{\ell}[t+1]$ in \eqref{LYphit1}, and the modulated symbols $\boldsymbol{\hat s}$.
Then, the noise element is derived by subtracting the noise-free signal from the received signal.
Consequently, the noise power ${\hat\sigma}^2[t+1]$ in iteration $t+1$ is simply estimated by computing the expectation of the noise power, as formulated in \eqref{sigma}.

\section{Convergence of the EM-based Channel Estimation}
\label{app:convergence}
The convergence of the proposed EM-based channel  estimation algorithm is discussed in this section, which is directly related to the performance of the BERD receiver.
The noise decomposition factor $\boldsymbol{w}$ is introduced to define the complete data, and the impact of $\boldsymbol{w}$ on the convergence rate and the convergence result of the EM-based channel estimation algorithm is clarified in Lemma \ref{le:convergence}.
\begin{lemma} \label{le:convergence}
\em{The choice of the noise decomposition factor $\boldsymbol{w}$ affects the convergence rate of the proposed EM-based estimation algorithm only, while it does not change the convergence results.}
\end{lemma}

\begin{proof}[\quad\it{Proof:}]
We first prove the impact of $\boldsymbol w$ on the convergence rate.
The EM algorithm utilizes the estimates of the previous iteration to update the new estimates of the unknown parameters by iterating between the E-step and the M-step. Thus, a mapping is defined as $\boldsymbol{\hat{\beta}}[t+1]=\mathfrak{F}(\boldsymbol{\hat{\beta}}[t])$,
where $\mathfrak{F}(\cdot)$ is a continuous function.
Note that $\mathfrak{F}(\cdot)$ can find a stationary point
when the EM algorithm converges, i.e., $\boldsymbol{\tilde{{\beta}}}=\mathfrak{F}(\boldsymbol{ \tilde {\beta}})$.\footnote{In our proposed algorithm, $\boldsymbol{\tilde{{\beta}}}$ may be a local optimal solution or a global optimal solution,
which depends on the initial.}
Then, the Taylor's series expansion of $\mathfrak{F}(\cdot)$ w.r.t. $\boldsymbol{\tilde {\beta}}$ can be expressed as \cite{ref[TWC27]}
\begin{align}
\label{eq:EM-F}
\mathfrak{F}(\boldsymbol{\hat{\beta}}[t])=\mathfrak{F}(\boldsymbol{\tilde{\beta}})
+\boldsymbol{U}(\boldsymbol{\hat{\beta}}[t]-\boldsymbol{\tilde{\beta}})
\end{align}
where $\boldsymbol{U}=\frac{\partial {\mathfrak{F}(\boldsymbol{\hat{\beta}}[t])}}{\partial \boldsymbol{\hat{\beta}}[t]}\bigg|
_{\boldsymbol{\hat{\beta}}[t]=\boldsymbol{\tilde{{\beta}}}}$.
By adopting the mapping function, \eqref{eq:EM-F} can be rewritten as
\begin{align}
\label{eq:h_r+1-h}
\boldsymbol{\hat{\beta}}[t+1]-\boldsymbol{\tilde{\beta}}
=\boldsymbol{U}(\boldsymbol{\hat{\beta}}[t]-\boldsymbol{\tilde{\beta}}).
\end{align}
From \cite{ref[TWC27]}, we know that the convergence rate $u_{c}$ of the EM algorithm
is defined as the largest eigenvalue ${\delta}_{\text{max}}$ of $\boldsymbol{U}$, i.e., ${u}_c={\delta}_{\text{max}}$.
In the following, we define $\boldsymbol{h}=\left[{h}_{0},{h}_{1},\ldots,{h}_{L-1}\right]^{\mathrm{T}}$
to simplify the expression of $\boldsymbol U$,
where ${h}_{\ell}={a}_{\ell}e^{\imath{\varphi}_{\ell}}$.
Then, the formula \eqref{LYat} and \eqref{LYphit1} are  rewritten as
\begin{align}
\label{eq:h=aphi}
{\hat h}_{\ell}[t+1]=
\frac{1}{{P}} \sum_{j\in\mathcal{I}_{N}}\hat {s}_{j-\ell}^{\ast}
\hat{z}_{\ell,j}[t].
\end{align}
Substituting the complete data in \eqref{LYcomplete} into \eqref{eq:h=aphi}, we have
\begin{align}
\label{eq:h_simplify}
{\hat h}_{\ell}[t+1]=
{\hat h}_{\ell}[t]+\frac{1}{{P}}\sum_{j\in\mathcal{I}_{N}}
{w}_{\ell}\Big(\hat {s}_{j-\ell}^{\ast}{r}_{j}- \hat {s}_{j-\ell}\sum_{\ell\in\mathcal{I}_{L}-1}
\hat {s}_{j-\ell}^{\ast}
{\hat h}_{\ell}[t]\Big).
\end{align}
By some manipulation, \eqref{eq:h_simplify} can be simplified as
\begin{align}
\label{eq:h_matrix}
\boldsymbol{\hat h}[t]=\Big(\boldsymbol{I}_L-\frac{1}{P}\big(\boldsymbol{w} \otimes \boldsymbol{1}_L^{\mathrm{T}} \big)
{\boldsymbol{\hat S}\boldsymbol{\hat S}^{\mathrm{H}}}\Big)^{\mathrm T}\Big(\boldsymbol{\hat h}[t+1]-\frac{1}{P}\big(\boldsymbol{w} \otimes \boldsymbol{1}_L^{\mathrm{T}} \big)
{\boldsymbol{\hat S}\boldsymbol{r}}\Big)
\end{align}
where $\boldsymbol{\hat S}\in \mathcal{S}^{L\times N}$ represents the transmitted symbol matrix, and the $\ell$th row of $\boldsymbol{\hat S}$ is the transmitted signal passing through the $\ell$th path, which has been determined in the soft-information detector and regenerator in Section \ref{sec:DEC}.
Substituting \eqref{eq:h_matrix} into \eqref{eq:h_r+1-h}, we have
\begin{align}
& \bigg(\boldsymbol{I}_L-\boldsymbol{U}\Big(\boldsymbol{I}_L-\frac{1}{P}\big(\boldsymbol{w} \otimes \boldsymbol{1}_L^{\mathrm{T}} \big)
{\boldsymbol{\hat S} \boldsymbol{\hat S}^{\mathrm{H}}}\Big)^{\mathrm T}\bigg)\boldsymbol{\hat h}[t+1] \\
&\quad\quad =(\boldsymbol{I}_L-\boldsymbol{U})\boldsymbol{\tilde{h}}-\frac{\boldsymbol{U}}{P}\Big(\boldsymbol{I}_L-\frac{1}{P}\big(\boldsymbol{w} \otimes \boldsymbol{1}_L^{\mathrm{T}} \big)
{\boldsymbol{\hat S} \boldsymbol{\hat S}^{\mathrm{H}}}\Big)^{\mathrm T}\big(\boldsymbol{w} \otimes \boldsymbol{1}_L^{\mathrm{T}} \big){\boldsymbol{\hat S} \boldsymbol{r}}.
\end{align}
From the mapping function $\boldsymbol{\hat h}[t+1]=\mathfrak{F}(\boldsymbol{\hat h}[t])$ and $\boldsymbol{U}=\frac{\partial {\mathfrak{F}(\boldsymbol{\hat h}[t])}}{\partial \boldsymbol{\hat h}[t]}\Big|
_{\boldsymbol{\hat h}[t]=\boldsymbol{\tilde{h}}}$, we obtain $\boldsymbol U$ as
\begin{align}
\label{eq:U}
\boldsymbol{U}=\boldsymbol{I}_L-\frac{1}{P}\big(\boldsymbol{w} \otimes \boldsymbol{1}_L^{\mathrm{T}} \big)
{\boldsymbol{\hat S} \boldsymbol{\hat S}^{\mathrm{H}}}.
\end{align}
The convergence rate $u_c$ is the largest eigenvalue of $\boldsymbol U$, which is related to $\boldsymbol w$.
Hence, we conclude that the noise decomposition factor $\boldsymbol w$ has impact on $u_c$.

To further illustrate the statement that the noise decomposition factor $\boldsymbol{w}$ has no impact on the convergence result, we first substitute the E-step in \eqref{LYcomplete} into \eqref{eq:h=aphi}, and \eqref{eq:h=aphi} can be rewritten as
\begin{align}
\label{eq:h_convergence}
{\hat h}_{\ell}[t+1]
=\frac{1}{{P}}\sum_{j\in\mathcal{I}_{N}}
{\hat s}_{j-\ell}^{\ast}
\Bigg(\bar{z}_{\ell,j}[t]+{w}_{\ell}\bigg({r}_{j}-\sum_{\ell\in\mathcal{I}_{L}-1}\bar{z}_{\ell,j}[t]\bigg)\Bigg).
\end{align}
By some manipulation, we have
\begin{align}
{\hat h}_{\ell}[t+1]
&=\frac{{\hat h}_{\ell}[t]}{{P}}\sum_{j\in\mathcal{I}_{N}}
|{\hat {s}}_{j-\ell}|^2
+\frac{{w}_{\ell}}{{P}}\sum_{j\in\mathcal{I}_{N}}{\hat {s}}_{j-\ell}^{\ast}\bigg({r}_{j}-\sum_{\ell\in\mathcal{I}_{L}-1}\bar{z}_{\ell,j}[t]\bigg) \\
&={\hat h}_{\ell}[t]+\frac{{w}_{\ell}}{{P}}
\sum_{j\in\mathcal{I}_{N}}{\hat s}_{j-\ell}^{\ast}\bigg({r}_{j}-\sum_{\ell\in\mathcal{I}_{L}-1}\bar{z}_{\ell,j}[t]\bigg).
\label{eq:h_converge_result}
\end{align}
Since it has been proved that, when the EM algorithm converges, we have $\|{\hat h}_{\ell}[t+1]-\hat{{h}}_{\ell}[t]\|^2\rightarrow 0$.
From \eqref{eq:h_converge_result}, we can see that the impact of the noise on the channel estimation becomes smaller, i.e., ${\hat s}_{j-\ell}^{\ast}({r}_{j}-\sum_{\ell\in\mathcal{I}_{L}-1}\bar{z}_{\ell,j}[t])\rightarrow 0$ as the iteration proceeds, which indicates that the choice of $\boldsymbol w$ has no impact on the convergence results.
\end{proof}

\begin{remark}
\em{Different from the general intuitions, the choice of $\boldsymbol w$ relevant to the complete data ${\hat z}_{\ell,j}$ in the E-step has no impact on the convergence result of the channel information estimation.
This crucial discovery guarantees the convergence and effectiveness of the proposed BERD receiver, which means no matter how to choose $\boldsymbol w$, the proposed scheme always converges to the same result.
Nevertheless, a better $\boldsymbol w$ can accelerate the convergence rate.}
\end{remark}

\section{Proofs of Lemma \ref{le:LLR-metric} and Theorem \ref{le:LLR-QAM}}
\label{app:LLR-QAM}
We first prove the LLR metric in $\mathrm{GF} (2)$ in Lemma \ref{le:LLR-metric}, which is used to prove the LLR of the syndrome APP stated in Theorem \ref{le:LLR-QAM}.

Considering two i.i.d. Bernoulli random variables ${\sf x}_1$ and ${\sf x}_2$, the probability of taking ${\sf x}_1\oplus {\sf x}_2=0$ is written as
\begin{align}
p({\sf x}_1\oplus {\sf x}_2=0)=p({\sf x}_1=0)p({\sf x}_2=0)+(1-p({\sf x}_1=0))(1-p({\sf x}_2=0))
\end{align}
where
\begin{align}
p({\sf x}_j=0)=\frac{e^{\mathcal{L}({\sf x}_j)}}{1+e^{\mathcal{L}({\sf x}_j)}}, \quad j\in\mathcal I_2.
\end{align}
Then, the LLR metric of ${\sf x}_1\oplus {\sf x}_2$ is derived as
\begin{align}
\mathcal{L}({\sf x}_1\oplus {\sf x}_2)
=\ln \frac{1+e^{\mathcal{L}({\sf x}_1)} e^{\mathcal{L}({\sf x}_2)}}{e^{\mathcal{L}({\sf x}_1)}+e^{\mathcal{L}({\sf x}_2)}}.
\end{align}
Furthermore, for the i.i.d. Bernoulli random variables ${\sf x}_j$, $j\in\mathcal I_N$,
$\mathcal{L}({\sf x}_1\oplus {\sf x}_2\oplus $ $\ldots \oplus {\sf x}_N)$ can de obtained by adopting the inductive methods
\begin{align}
\mathcal{L}({\sf x}_1\oplus {\sf x}_2\oplus \ldots \oplus {\sf x}_N)
=\ln \frac{\prod_{j\in\mathcal{I}_N}\left(e^{\mathcal{L}({\sf x}_j)}+1\right) + \prod_{j\in\mathcal{I}_N}\left(e^{\mathcal{L}({\sf x}_j)}-1\right)}
{\prod_{j\in\mathcal{I}_N}\left(e^{\mathcal{L}({\sf x}_j)}+1\right) - \prod_{j\in\mathcal{I}_N}\left(e^{\mathcal{L}({\sf x}_j)}-1\right)}.
\end{align}
By utilizing the function $\tanh\frac{1}{2}{{\sf x}_j}=\frac{e^{{\sf x}_j}-1}{e^{{\sf x}_j}+1}$,
$\mathcal{L}({\sf x}_1\oplus {\sf x}_2\oplus \ldots \oplus {\sf x}_N)$ is rewritten as
\begin{align}
\label{eq:tanh-1}
\mathcal{L}({\sf x}_1\oplus {\sf x}_2\oplus \ldots \oplus {\sf x}_N)
&=\ln \frac{1+\prod_{j\in\mathcal{I}_N}\tanh\frac{1}{2}{\mathcal{L}({\sf x}_j)}}
{1-\prod_{j\in\mathcal{I}_N}\tanh\frac{1}{2}{\mathcal{L}({\sf x}_j)}} \\
&=2\tanh^{-1}\prod_{j\in\mathcal{I}_N}\tanh \frac{1}{2}{\mathcal{L}({\sf x}_j)}.
\end{align}
Hence, we obtain the LLR metric in Lemma \ref{le:LLR-metric}.

To further prove the LLR of the syndrome APP stated in Theorem \ref{le:LLR-QAM}, we first derive the posterior probability LLR of the coded bit ${\sf {c}}_{j,g}$, which is denoted by
\begin{align}
\mathcal{L}({\sf {c}}_{j,g}|{r}_{j},\boldsymbol{{s}}_{j-L+1}^{j-1};{\boldsymbol{\beta}})
&={\lambda}_{j,g}^{\text{out}},
\quad j\in\mathcal{I}_N, \quad g\in\mathcal{I}_{\log|\mathcal S|}.
\label{QAMLLR}
\end{align}
Since we assume perfect synchronization, the relation between the codeword $\boldsymbol{\sf{\tilde{c}}}$ and the coded bits ${\sf c}_{j,g}$, $j\in\mathcal{I}_N$, $g\in\mathcal{I}_{\log|\mathcal S|}$, is $[\tilde{\sf c}_1,\tilde{\sf c}_2,\ldots,\tilde{\sf c}_n]^\mathrm{T}=[{\sf c}_{1,1},{\sf c}_{1,2},\ldots,{\sf c}_{1,\log|\mathcal S|},{\sf c}_{2,1},\ldots,{\sf c}_{N,\log|\mathcal S|}]^\mathrm{T}$.
In addition, $\boldsymbol{\psi}=\boldsymbol{\lambda}^{\text{out}}$.
Given the modulation $\eta$ and a $(n,q)$ linear block code $\zeta$, from Lemma \ref{le:LLR-metric}, \eqref{eq:syndrome}, and \eqref{QAMLLR}, the LLR of the syndrome APP of the $i$th parity-check bit is obtained by
\begin{align}  \label{gamaQAM}
{\gamma}_{i}^{\eta,\zeta}
&=\mathcal{L} \bigg({\tilde{\sf c}_{\pi_{i}^{\eta,\zeta}(1)}}  \oplus  \cdots  \oplus {\tilde{\sf c}_{{\pi_{i}^{\eta,\zeta}(N_{i})}}}|\boldsymbol{r},\boldsymbol{s}^{\eta,\zeta};\boldsymbol{\beta}^{\eta,\zeta}\bigg) \\
&= 2\tanh^{-1}\Bigg(\prod_{\tau\in\mathcal I_{N_i}}\tanh\Big(\frac{1}{2}{\psi}_{\pi_{i}^{\eta,\zeta}(\tau)}\Big)\Bigg),
\quad i\in\mathcal{I}_{n-q}.
\end{align}
Therefore, Theorem \ref{le:LLR-QAM} is concluded.

\bibliographystyle{plainnat}

\bibliography{BERD_iterative}

\end{document}